\documentclass[a4paper, cleveref, autoref, thm-restate]{lipics-v2021}

\hideLIPIcs

\title{
Net Occurrences in Fibonacci and Thue-Morse Words
}

\author{Peaker Guo}
{
School of Computing and Information Systems,
The University of Melbourne, 
Parkville, Australia
\and
\url{www.peakerguo.com}
}{zifengg@student.unimelb.edu.au}
{https://orcid.org/0000-0002-9098-1783}{}

\author{Kaisei Kishi}
{
Department of Information Science and Technology, 
Kyushu University,
Fukuoka, Japan 
}{kishi.kaisei.216@s.kyushu-u.ac.jp}{}{}

\authorrunning{
P.~Guo and K.~Kishi
}

\Copyright{Peaker Guo and Kaisei Kishi}



\funding{
\textit{Peaker Guo}:
Supported by an Australian Government Research Training Program Scholarship.
}

\acknowledgements{
The authors thank Hideo Bannai and the organizers of 
the StringMasters workshop at CPM 2024
for initiating the collaboration between the authors during the workshop.
The authors also thank Shunsuke Inenaga and William Umboh for their helpful advice. 
}

\nolinenumbers

\EventEditors{Paola Bonizzoni and Veli M\"{a}kinen}
\EventNoEds{2}
\EventLongTitle{36th Annual Symposium on Combinatorial Pattern Matching (CPM 2025)}
\EventShortTitle{CPM 2025}
\EventAcronym{CPM}
\EventYear{2025}
\EventDate{June 17--19, 2025}
\EventLocation{Milan, Italy}
\EventLogo{}
\SeriesVolume{331}
\ArticleNo{17}

\usepackage{amsmath}
\usepackage{amssymb}
\usepackage{mathtools}
\usepackage[noadjust]{cite}
\usepackage{caption}
\usepackage{subcaption}

\usepackage[capitalise, noabbrev]{cleveref}
\usepackage[dvipsnames]{xcolor}

\usepackage{xspace}
\xspaceaddexceptions{]\}}

\newcommand{\tm}{\ensuremath{\mathcal{T}}\xspace} 
\newcommand{\ltm}{\ensuremath{\tau}\xspace}
\newcommand{\flip}[1]{\ensuremath{\overline{#1}}\xspace}
\newcommand{\fac}[2]{\ensuremath{\mathcal{F}_{#2}^{#1}}\xspace}
\newcommand{\nextfac}[1]{\ensuremath{\left(#1\right)^-}\xspace}
\newcommand{\emptyseq}{\ensuremath{( \; )}\xspace}

\newcommand{\removed}[1]{}

\usepackage{multirow}

\begin{document}

\maketitle 

\begin{abstract}
A \emph{net occurrence} of a repeated string in a text
is an occurrence with unique left and right \emph{extensions},
and the \emph{net frequency} of the string is the 
number of its net occurrences in the text.
Originally introduced for applications in Natural Language Processing, 
net frequency has recently gained attention for 
its algorithmic aspects.
Guo et al.~[CPM 2024] and Ohlebusch et al.~[SPIRE 2024]
focus on its computation in the offline setting, while
Guo et al.~[SPIRE 2024], Inenaga~[arXiv 2024],
and Mieno and Inenaga~[CPM 2025]
tackle the online counterpart.
Mieno and Inenaga also characterize 
 net occurrences in terms of 
the minimal unique substrings of the text.
Additionally, Guo et al.~[CPM 2024] initiate the study of 
 net occurrences in Fibonacci words 
to establish a lower bound on the asymptotic running time of algorithms.
Although there has been notable progress in algorithmic developments and some initial combinatorial insights, 
the combinatorial aspects of net occurrences have yet to be thoroughly examined. 
In this work, we make two key contributions.
First, we confirm the conjecture that 
each Fibonacci word contains exactly three net occurrences.
Second, we show that 
each Thue-Morse word contains exactly nine net occurrences. 
To achieve these results,
we introduce the notion of \emph{overlapping net occurrence cover},
which narrows down the candidate net occurrences in any text.
Furthermore,
we provide a precise characterization of
occurrences of Fibonacci and Thue-Morse words of smaller order, 
offering structural insights that may have independent interest
and potential applications in algorithm analysis and combinatorial properties of these words.
\end{abstract}

\clearpage

\clearpage

\section{Introduction}

The work by Axel Thue at the beginning of the 20th century
marked the beginning of the field of combinatorics on words~\cite{journal/ejc/2007/berstel}.
Central to the field are two key objects that have attracted extensive research:
Fibonacci words and Thue-Morse words~\cite{book/1997/lothaire}.
These objects are remarkable 
for their rich combinatorial properties and applications in seemingly unrelated fields 
beyond combinatorics on words.
Fibonacci words, for instance, have been used to establish lower bounds 
and analyze behaviors of string algorithms~\cite{journal/mst/2022/inoue},
while Thue-Morse words appear in diverse areas such as 
 group theory, physics, and even chess~\cite{conf/seta/1998/allouche}.
They have also been used to prove properties related to repetitiveness measures~\cite{conf/spire/2021/bannai, conf/spire/2020/kutsukake, conf/cwords/2023/dolce, journal/tit/2021/navarro, conf/cpm/2025/bannai}.

Another key aspect of combinatorics on words involves identifying significant strings in a text.
Different definitions of significance lead to different problem formulations.
These significant strings could be 
repetitions~\cite{journal/tcs/2009/crochemore},
tandem repeats~\cite{journal/jcss/2004/gusfield}, or
runs~\cite{journal/siamcomp/2017/bannai}.
There is also a rich literature on the study of these significant strings 
in Fibonacci and Thue-Morse words~\cite{journal/dam/1989/brlek, journal/ipl/1995/droubay, journal/tcs/1997/iliopoulos, journal/tcs/1999/fraenkel}.
For many applications, frequency serves as a basis for significance measure. 
However, frequency alone can be misleading, as it may be 
inflated by occurrences of longer repeated strings.
Consider the text 
\texttt{the\textvisiblespace theoretical\textvisiblespace theme} as an example.
The string \texttt{the} is the most frequent string of length three, but this is due to the fact that
two of its occurrences are contained by the longer repeated string \texttt{\textvisiblespace the}.

To address this issue, Lin and Yu~\cite{journal/jise/2001/lin, journal/ijclclp/2004/lin} 
introduced the notion of \emph{net frequency} (NF),
motivated by Natural Language Processing tasks.
As reconceptualized by Guo et al.~\cite{conf/cpm/2024/guo},
a \emph{net occurrence} of a repeated string
in a text is an occurrence with unique left and right extensions,
and the NF of the string is the number of its net occurrences
in the text.
In the earlier example, 
only the first occurrence 
of \texttt{the}
is a net occurrence,
reflecting the only occurrence 
that is not contained by a longer repeated string.

There has been a recent surge of interest in the computation of NF.
Guo et al.~\cite{conf/cpm/2024/guo} and Ohlebusch et al.~\cite{conf/spire/2024/ohlebusch}
focus on the offline setting,
while Guo et al.~\cite{conf/spire/2024/guo}, 
Inenaga~\cite{DBLP:journals/corr/abs-2410-06837},
and 
Mieno and Inenaga~\cite{conf/cpm/2025/mieno}
extend the computation to the online setting.
Mieno and Inenaga also characterize 
 net occurrences in terms of 
the minimal unique substrings of the text.
Additionally, Guo et al.~\cite{conf/cpm/2024/guo}  study 
 net occurrences in Fibonacci words 
to establish a lower bound on the asymptotic running time of algorithms.
Despite these advances,
the combinatorial aspect of net occurrences has yet to be thoroughly investigated.
It has been shown that there are \emph{at least} three net occurrences 
in each Fibonacci word~\cite{conf/cpm/2024/guo}.
However, proving that these are the \emph{only} three 
is more challenging and
was only conjectured.
Meanwhile, the net occurrences in each Thue-Morse word had not been 
investigated before---both of which we address in this work.

\subparagraph*{Our results.}
In this work,
our main contribution is twofold.
First, we confirm the conjecture by 
Guo et al.~\cite{conf/cpm/2024/guo} that 
there are exactly three net occurrences
in each Fibonacci word (\cref{thm:fib-only-net-occ}).
Second, we show that 
there are exactly nine net occurrences 
in each Thue-Morse word (\cref{thm:tm-only-net-occ}).
To achieve these results, 
 we first introduce the concept of an \emph{overlapping net occurrence cover}, 
which drastically reduces the number of occurrences that need to be examined 
when proving certain net occurrences are the only ones (\cref{thm:super-occ-bnso}). 
Additionally, we provide a precise characterization of 
occurrences of smaller-order Fibonacci and Thue-Morse words 
(\cref{thm:fib-occ-inductive} and \cref{thm:tm-occ}).
These findings 
could also be of independent interest,
providing tools and insights
for analyzing algorithms and exploring the combinatorial properties of these words.
For example, they lead to methods to count the smaller-order occurrences (\cref{thm:fib-occ-num} and \cref{thm:tm-occ-num}).

\subparagraph*{Other related work.}
Occurrences of Fibonacci and Thue-Morse words of smaller order have been previously studied. 
For Fibonacci words, these occurrences have been shown to be related to the Fibonacci representation of positive integers~\cite{journal/tcs/1997/iliopoulos, journal/tcs/2006/rytter}.
For Thue-Morse words, these occurrences have been investigated using the binary representation of numbers and properties of the compact directed acyclic word graph (CDAWG) of each Thue-Morse word~\cite{journal/jda/2012/radoszewski}.
We emphasize that our work addresses occurrences of Fibonacci and Thue-Morse words of smaller order from a different angle than prior work: we provide a recurrence relation that precisely characterizes the occurrences, bypassing the need for other representations.

\section{Preliminaries}

\subparagraph*{Strings.}
Throughout, we consider the binary alphabet $\Sigma := \{\texttt{a}, \texttt{b}\}$.
A \emph{string} is an element of $\Sigma^*$.
The length of a string $S$ is denoted as $|S|$.
Let $\epsilon$ denote the \emph{empty string} of length 0.
We use $S[i]$ to denote the $i^{\text{th}}$ character of a string $S$.
Let~$[n]$ denote the set $\{1,2,\ldots,n\}$.
Let $S \ T$ be the \emph{concatenation} of two strings,~$S$ and~$T$.
A substring of a string~$T$ of length $n$, 
starting at position~$i \in [n]$ 
and ending at position~$j \in [n]$, is written as $T[i \ldots j]$. 
A substring $T[1\ldots j]$ is called a \emph{prefix} of~$T$, 
while $T[i \ldots n]$ is called a \emph{suffix} of $T$.
A substring $S$ of $T$ is a \emph{proper} substring if $S \neq T$.
An \emph{occurrence} in the text~$T$ of length $n$ is a pair of 
starting and ending positions $(i, j) \in [n] \times [n]$.
We say $(i, j)$ is an \emph{occurrence of string~$S$} if $S = T[i \ldots j]$,
and $i$ is an occurrence of $S$ if $S = T[i \ldots i + |S| - 1]$.
An occurrence $(i', j')$ is a \emph{sub-occurrence} of $(i, j)$
if $i \leq i' \leq j' \leq j$.
An occurrence $(i, j)$ is a \emph{super-occurrence} of $(i', j')$ if 
$(i', j')$ is a sub-occurrence of $(i, j)$.
Moreover, 
$(i', j')$
is a \emph{proper} sub-occurrence (or super-occurrence) of $(i, j)$ if $(i', j')$ is a sub-occurrence (or super-occurrence) of $(i, j)$ and $i \neq i'$ or $j' \neq j$.
Two occurrences $(i, j)$ and $(i', j')$ \emph{overlap} if 
there exists a position $k$ such that $i \leq k \leq j$ and $i'\leq k \leq j'$.
For a non-empty string $S$, 
a sequence of non-empty strings $\mathcal{F} = (  x_k )^m_{k=1} = (x_1, x_2, \ldots, x_m)$
is referred to as a \emph{factorization} of $S$ if
$S = x_1 \ x_2 \ \cdots x_m $.
Each string $x_k$ is called a \emph{factor} of $\mathcal{F}$.
The \emph{size} of $\mathcal{F}$, denoted by $|\mathcal{F}|$, 
is the number of factors in the factorization.

\subparagraph*{Net frequency and net occurrences.}
In a text $T$,
the net frequency (NF) of a unique string in $T$ is defined to be zero.
The NF of a repeated string is the number of \emph{net occurrences} in $T$.

\begin{definition}[Net occurrence~\cite{conf/cpm/2024/guo}]
In a text $T$,
an occurrence $(i, j)$ is a \emph{net occurrence} if
the corresponding string 
$T[i   \ldots j  ]$ is repeated,
while both left extension 
$T[i-1 \ldots j  ]$ and right extension
$T[i   \ldots j+1]$ are unique.
When $i=1$, $T[i-1 \ldots j  ] $ is assumed to be unique;
when $j=|T|$, $T[i   \ldots j+1] $ is assumed to be unique.
\end{definition}

For an occurrence $(i, j)$ in text $T$,
we refer to $T[i-1]$ and $T[j+1]$ as the left and right extension \emph{characters} 
of $(i, j)$, respectively.
For a string $S$ occurring in $T$, 
we say $x, y \in \Sigma$ are  left and right extension  \emph{characters} 
of $S$ if both strings $xS$ and $Sy$ also occur in $T$.

\subparagraph*{Fibonacci words.}
Let $F_i$ denote the (finite) \emph{Fibonacci word of order $i$}
where~$F_1 := \texttt{b}, F_2 := \texttt{a}$, and 
$F_i := F_{i-1} \ F_{i-2}$ for each $i \geq 3$. 
Let $f_i := |F_i|$ be the length of the Fibonacci word of order $i$, which is also 
the $i^{\text{th}}$ Fibonacci number. 
We next review two useful results on $F_i$.

\begin{lemma}[\cite{journal/ipl/1995/droubay}]\label{thm:f-i-f-i}
$F_{i}$ only occurs twice in $F_i \ F_i$.
\end{lemma}

\begin{lemma}[\cite{journal/tit/2021/navarro}]\label{thm:no-aaa}
The strings \texttt{aaa} and \texttt{bb} do not occur in $F_i$.
\end{lemma}

The following result can be readily
derived by repeatedly applying the definition of $F_i$.

\begin{observation}\label{thm:fib-basic-fac}
For $1 \leq k \leq i$,
there is a factorization of $F_i$ where each factor is either $F_{k}$ or $F_{k+1}$.
\end{observation}

\noindent
For example,
for $k = i-2 \ldots i-5$,
we have the following factorizations:
$ 
  F_i
= F_{i-1} \ F_{i-2} 
= F_{i-2}  \ F_{i-3} \ F_{i-2}
= F_{i-3} \ F_{i-4} \ F_{i-3} \ F_{i-3} \ F_{i-4}
= F_{i-4} \ F_{i-5} \ F_{i-4} \ F_{i-4} \ F_{i-5} \ F_{i-4} \ F_{i-5} \ F_{i-4}
$.

\subparagraph*{Thue-Morse words.}
For a binary string $S$,
let $\flip{S}$ denote the string obtained by
simultaneously replacing each \texttt{a} with \texttt{b} 
and each \texttt{b} with \texttt{a}.
Let $\tm_i$ be the (finite) \emph{Thue-Morse word of order $i$}
where $\tm_1 := \texttt{a}$ and
$\tm_i := \tm_{i-1} \flip{ \tm_{i-1}}$
for each $i \geq 2$.
Let $\ltm_i := |\tm_i| = 2^{i-1}$ be the length of the Thue-Morse word of order $i$.
We next review two properties of each $\tm_i$.

\begin{lemma}[Overlap-free~\cite{book/1997/lothaire}]\label{thm:tm-overlap-free}
$\tm_i$ has no overlapping occurrences of the same string.
\end{lemma}

\begin{lemma}[Cube-free~\cite{book/1997/lothaire}]\label{thm:tm-cube-free}
$\tm_i$ does not contain any string of the form $xxx$ where $x$ is a non-empty string.
\end{lemma}

The following result can be directly derived by repeatedly applying the definition of $\tm_i$.

\begin{observation}\label{thm:tm-basic-fac}
For each $i \geq 2$ and $1 \leq j \leq i$, 
there is a factorization of $\tm_i$ where each factor is either $\tm_{i-(j-1)}$ or $\flip{\tm_{i-(j-1)}}$.
\end{observation}

\noindent
For example, for $2 \leq j \leq 3$,
$
\tm_i 
= \tm_{i-1} \ \flip{\tm_{i-1}}
= \tm_{i-2} \ \flip{\tm_{i-2}}  \ \flip{\tm_{i-2}} \ \tm_{i-2}
$.
\cref{fig:tm-occ} illustrates larger value of $j$.
Also note that this result is analogous to \cref{thm:fib-basic-fac}.

\section{Overlapping Net Occurrence Cover}\label{sec:onoc}

This section lays the foundation 
to prove the main results of this paper 
in the subsequent sections.
Specifically, we aim to develop tools to show that certain net occurrences are the only ones in a text. 
To achieve this, 
we first provide two characteristics 
for non-net occurrences.
The proofs in this section are presented in \cref{appx:onoc}.

\begin{restatable}{observation}{nonnetsuper}\label{thm:non-net-super}
In a text $T$,
if an occurrence $(s,e)$ is a proper super-occurrence of a net occurrence,
then $(s,e)$ is not a net occurrence.   
\end{restatable}

\begin{restatable}{observation}{nonnetsub}\label{thm:non-net-sub}
In a text $T$,
if an occurrence $(s,e)$ is a proper sub-occurrence of a net occurrence,
then $(s,e)$ is not a net occurrence.     
\end{restatable}

\begin{remark}
In a text, both a string and its substring
can have positive NF, for example, in \texttt{abaababaabaab}, both \texttt{abaaba} and  \texttt{abaab} have positive NF.
However, this relationship does not hold for an occurrence and its sub-occurrence,
as shown in the above two observations. 
\end{remark}

To show that a given set of net occurrences in $T$
are the only ones in $T$,
the above two observations allow us to ignore
any occurrence that is either a sub-occurrence or 
a super-occurrence of a net occurrence.
To fully use these two observations,
we focus on the case when
the given net occurrences ``overlap'' one another and collectively ``cover'' the text.
Consequently, the only occurrences that need to be explicitly examined 
are the super-occurrences of those corresponding to the ``overlapping regions''
of these net occurrences. 
To formalize this, we introduce the following definition and lemma.

\begin{definition}[ONOC and BNSO]\label{def:onoc}
Consider a text $T$ and 
a set of $c$ net occurrences in $T$:
$\mathcal{C} = \{(i_1, j_1), (i_2, j_2), \ldots, (i_c, j_c) \}$.
We say $\mathcal{C}$ is an
\emph{overlapping net occurrence cover (ONOC)} of $T$ if
$i_1 = 1$, $i_{k+1} \leq j_k$ for $1 \leq k \leq c-1$, and $j_c = n$.
Each occurrence in the set 
$\{ (i_2, j_1), (i_3, j_2), \ldots (i_{c}, j_{c-1}) \}$ 
is a \emph{bridging net sub-occurrence (BNSO)} of $\mathcal{C}$.
\end{definition}

\noindent
An example of \cref{def:onoc}
is shown in \cref{fig:onoc-example}.

\begin{figure}[t]
\centering
\includegraphics[width=0.4\linewidth]{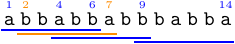}
\caption{
An example for \cref{def:onoc}.
The set
$\{(1,6),(4,9),(9,14)\}$
is an ONOC,
with each of its net occurrences underlined in blue;
$\{(4,6),(9,9)\}$ is the corresponding set of BNSOs.
Note that $(2, 7)$ is a net occurrence outside of this ONOC,
underlined in orange.
}
\label{fig:onoc-example}
\end{figure}

\begin{restatable}{lemma}{superoccbnso}\label{thm:super-occ-bnso}
For a text $T$,
if there exists an ONOC $\mathcal{C}$ of $T$ such that 
$\mathcal{C}$ does not contain all the net occurrences in $T$, then each
net occurrence in $T$ outside of $\mathcal{C}$  must be  a  
super-occurrence of $(i-1, j+1)$,
where $(i, j)$ is  a BNSO of $\mathcal{C}$.    
\end{restatable}

In the example in \cref{fig:onoc-example},
note that net occurrence $(2,7)$ 
is indeed a super-occurrence of 
$(4-1, 6+1)$, where $(4, 6)$ is a BNSO.

In \cref{sec:fib-net-occ} and \cref{sec:tm-net-occ},
we apply \cref{thm:super-occ-bnso}
in three steps.
First, for a Fibonacci or Thue–Morse word, we show that an ONOC exists. 
Next, we examine the set of BNSOs of the ONOC.
Finally, we prove that no super-occurrence
of $(i-1,j+1)$ (where $(i, j)$ is a BNSO) is a  net occurrence,
thus concluding that the ONOC already contains
all the net occurrences in the text.

\section{
Occurrences of Fibonacci Words of Smaller Order
}\label{sec:occ-fib-smaller-ord}

We study the occurrences of $F_{i-j}$ in $F_i$
for appropriate $i$ and $j$.
These results will help us prove the only net occurrences in $F_i$ in \cref{sec:fib-net-occ} and  may also be of independent interest.

When $j = 1$, with $F_i = F_{i-1} \ F_{i-2}$,
we have one occurrence of $F_{i-1}$ at position 1.
The following result shows that this is the only one.

\begin{lemma}[\cite{journal/tit/2021/navarro}]\label{thm:f-i-1}
$F_{i-1}$ only occurs at position 1 in $F_i$ for $i \geq 3$.
\end{lemma}

The two factorizations 
in the following result
reveal three occurrences of $F_{i-2}$ in $F_i$.

\begin{observation}[\cite{conf/cpm/2024/guo}]\label{thm:f-i-fac}
For each $i \geq 6$,
\begin{align}
F_i &= F_{i-2} \ F_{i-3} \ F_{i-2} \label{eq:fac-1}  \\
F_i &= F_{i-2} \ F_{i-2} \ F_{i-5} \ F_{i-4} \label{eq:fac-2} .
\end{align}    
\end{observation}

\noindent
The following result confirms that these are the only three.

\begin{lemma}[\cite{journal/tit/2021/navarro}]\label{thm:f-i-2}
$F_{i-2}$ only occurs at positions 
$1$, $f_{i-2}+1$, and $f_{i-1}+1$ in $F_i$ for $i \geq 6$.
\end{lemma}

\noindent
We next provide the result when $j=3$ and $i \geq 7$.

\begin{lemma}\label{thm:f-i-3}
$F_{i-3}$ only occurs at positions 
$1$, $f_{i-3}+1$, $f_{i-2}+1$, and $f_{i-1}+1$ in $F_i$.
\end{lemma}

\begin{proof}
From \cref{thm:f-i-2},
notice that the second occurrence of $F_{i-2}$ 
follows immediately after the first occurrence,
and the second and the third occurrences of $F_{i-2}$ overlap.
Then, based on \Crefrange{eq:fac-1}{eq:fac-2},
we consider the following three cases.

\begin{enumerate}[{Case} 1]
\item
$F_{i-3}$ occurs within $F_{i-2}$.
From \cref{thm:f-i-1}, $F_{i-3}$ only occurs at position 1 in $F_{i-2}$.
Thus, using \cref{thm:f-i-2},
the only occurrences of $F_{i-3}$ within $F_{i-2}$ in $F_i$
are at positions $1$, $f_{i-2}+1$, and $f_{i-1}+1$.

\item
$F_{i-3}$ occurs across the boundary of $F_{i-2} \ F_{i-3}$.
Again from \cref{thm:f-i-2},
the only occurrences of $F_{i-3}$ within $F_{i-2} \ F_{i-3} = F_{i-1}$
are at positions $1$, $f_{i-3}+1$, and $f_{i-2}+1$.
Note that the occurrence at position $f_{i-3}+1$ is the boundary-crossing one:
we apply \cref{eq:fac-2} on $F_{i-1}$ and obtain
$F_{i-2} \ F_{i-3}  = F_{i-3} \ F_{i-3} \ F_{i-6} \ F_{i-5}$.

\item
$F_{i-3}$ occurs across the boundary of $F_{i-3} \ F_{i-2}$.
Note that 
$ F_{i-3} \ F_{i-2}  = F_{i-3} \ F_{i-3} \ F_{i-4}$.
Using \cref{thm:f-i-f-i},
$F_{i-3}$ does not occur in $F_{i-3} \ F_{i-3}$ and thus does not occur across the boundary of 
$F_{i-3} \ F_{i-2}$.
\qedhere
\end{enumerate} 
\end{proof}

We now present the main result of the section,
illustrated in \cref{fig:fib-occ}.
Before that, we define the following.
For a set of integers $A$ and another integer $i$,
$A \oplus i$ denotes the set 
$\{ a + i \ : \ a \in A \}$.
We write $\max(A)$ for the maximum element of set $A$.

\begin{figure}[t]
\centering
\includegraphics[width=0.55\linewidth]{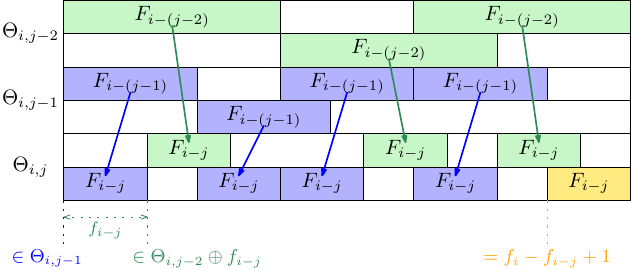}
\caption{An illustration of \cref{thm:fib-occ-inductive} when $j = 4$.
Each row depicts a factorization of $F_i$ with relevant factors highlighted in colors.
The top two, middle two, and bottom two rows correspond to sets 
$\Theta_{i,j-2}$, $\Theta_{i,j-1}$ and $\Theta_{i,j}$, respectively.
Each green and blue occurrence of $F_{i-j}$ is introduced by 
an occurrence of $F_{i-(j-2)}$ and $F_{i-(j-1)}$, respectively.
The yellow occurrence is the rightmost one.
}
\label{fig:fib-occ}
\end{figure}

\begin{theorem}\label{thm:fib-occ-inductive} 
Let $\Theta_{i,j}$ denote the set of the starting positions of the occurrences of $F_{i-j}$ in $F_{i}$.
Then, $\Theta_{i,0} = \Theta_{i,1} = \{1\}$, 
and for $2 \leq j \leq i - 4$,
\begin{itemize}
\item 
when $j$ is even, 
$ \Theta_{i,j} = \Theta_{i,j-1} \cup (\Theta_{i,j-2} \oplus f_{i-j}) \cup \{f_i - f_{i-j} + 1 \}$,
where 
the three sets in the union
are mutually disjoint,
and $\max(\Theta_{i,j}) = f_i - f_{i-j} + 1 $;
\item 
when $j$ is odd,
$\Theta_{i,j} = \Theta_{i,j-1} \cup (\Theta_{i,j-2} \oplus f_{i-j})$
where
the two sets in the union are disjoint,
and $\max(\Theta_{i,j}) = f_i - f_{i-(j-1)} + 1$.
\end{itemize}
\end{theorem}

\begin{proof}
We proceed by induction on $j$.

\textbf{Base cases.}
When $j = 0$, we have $\Theta_{i,0} = \{1 \}$ trivially. 
When $j = 1$, it follows from \cref{thm:f-i-1} that $\Theta_{i,1} = \{1 \}$.
When $j = 2$, we obtain $\Theta_{i,2} = \{1, f_{i-2} +1, f_{i-1} + 1 \}$ by \cref{thm:f-i-2}. 
Thus, $\Theta_{i,2} = \Theta_{i,1} \cup (\Theta_{i,0} \oplus f_{i-2}) \cup \{ f_i - f_{i-2} + 1 \}$, 
 where the three sets are mutually disjoint, and 
 $\max(\Theta_{i,2}) = f_{i-1} + 1 = f_i - f_{i-2} + 1$.
Next, when $j = 3$, we have 
$\Theta_{i,3} = \{1, f_{i-3} +1, f_{i-2} +1, f_{i-1} + 1 \}$ by \cref{thm:f-i-3}.
Hence, $\Theta_{i,3} = \Theta_{i,2} \cup (\Theta_{i,1} \oplus f_{i-3})$, 
$\Theta_{i,2} \cap (\Theta_{i,1}  \oplus f_{i-3}) = \emptyset$,  
and $\max(\Theta_{i,3}) = f_{i-1} + 1 = f_i - f_{i-(3-1)} + 1$.

\textbf{Inductive step.}  
For each $4 \leq k \leq i-4$, 
assume the claim holds for $j=k-2$ and $k-1$,
and we now prove the claim for $j = k$. 
We first prove the claim for even $k$.
Define $\Lambda_{i, k} := \Theta_{i,k-1} \cup (\Theta_{i,k-2} \oplus f_{i-k}) \cup \{f_i - f_{i-k} + 1 \}$.
We aim to show that $\Theta_{i,k} = \Lambda_{i, k}$ by showing
$\Lambda_{i, k} \subset \Theta_{i,k}$ and 
$\Theta_{i,k} \subset \Lambda_{i, k}$.
Before we proceed, note that
$F_{i-(k+2)}, F_{i-(k+1)}, F_{i-k}, F_{i-(k-1)}, F_{i-(k-2)}$
are consecutive Fibonacci words of \emph{increasing} orders.

To prove that
$\Lambda_{i, k} \subset \Theta_{i,k}$,
we will show that each set in the union defining $\Lambda_{i, k}$
is contained in $\Theta_{i,k}$.
Note that $\Theta_{i,k-1} \subset \Theta_{i,k}$
because $F_{i-k}$ is a prefix of $F_{i-(k-1)} = F_{i-k} \ F_{i-(k+1)}$.
Next, we have $(\Theta_{i,k-2} \oplus f_{i-k}) \subset \Theta_{i,k}$
because $F_{i-k}$ occurs at position $f_{i-k} + 1$ in 
$
F_{i-(k-2)} 
= F_{i-k} \ F_{i-k} \ F_{i-(k+3)} \ F_{i-(k+2)}
$
where this factorization can be derived similarly to \cref{eq:fac-2}.
Lastly, 
by the induction hypothesis on $j = k-2$,
we have $f_i - f_{i-(k-2)} + 1 \in \Theta_{i, k-2}$
and it is the rightmost occurrence of $F_{i-(k-2)}$ in $F_i$.
Consider the factorization $F_{i-(k-2)} = F_{i-k} \ F_{i-(k+1)} \ F_{i-k}$,
which can be derived similarly to \cref{eq:fac-1}.
Note that the rightmost occurrence of $F_{i-k}$ in $F_{i-(k-2)}$
is 
at position 
$
  (f_i - f_{i-(k-2)} + 1) + (f_{i-k} + f_{i-(k+1)})
= (f_i - (f_{i-k} + f_{i-k} + f_{i-(k+1)}) + 1) + (f_{i-k} + f_{i-(k+1)})
= f_i - f_{i-k} + 1
$
in $F_i$.
Thus, $f_i - f_{i-k} + 1 \in \Theta_{i,k}$.

Next, we prove $\Theta_{i,k} \subset \Lambda_{i, k}$
by showing that each occurrence of $F_{i-k}$ is in $\Lambda_{i, k}$.
By \cref{thm:fib-basic-fac}, 
there is a factorization of $F_i$ where each factor is either $F_{i-(k-1)}$ or $F_{i-k}$.
We now examine the occurrences of $F_{i-k}$ based on this factorization. 
First, when there is an occurrence of $F_{i-k}$ within $F_{i-(k-1)}$,
this occurrence is in $\Theta_{i,k-1} \subset \Lambda_{i,k}$.
Next, when there is an occurrence of $F_{i-k}$ (underlined)  across the boundary of 
$
  F_{i-(k-1)} \ F_{i-(k-1)} 
= F_{i-k} \ \underline{F_{i-k}} \  F_{i-(k+3)} \ F_{i-(k+2)} \ F_{i-(k+1)} 
$
or across the boundary of 
$
F_{i-(k-1)} \ F_{i-k} = F_{i-(k-2)}
= F_{i-k} \ \underline{F_{i-k}} \ F_{i-(k+3)} \ F_{i-(k+2)}
$,
then, 
by the fact that 
$F_{i-k}$ only occurs at positions 
$1$, $f_{i-k}+1$, and $f_{i-(k-1)}+1$ in $F_{i-(k-2)}$ 
(a direct generalization of \cref{thm:f-i-2}),
 this occurrence of $F_{i-k}$ must be in
$(\Theta_{i,k-2} \oplus f_{i-k}) \subset \Lambda_{i,k}$.
Finally, observe that
there does not exist an occurrence of $F_{i-k}$  across the boundary of 
$
F_{i-k} \ F_{i-(k-1)} = F_{i-k} \ F_{i-k} \ F_{i-(k-1)}
$
or across the boundary of 
$
F_{i-k} \ F_{i-k}
$,
because otherwise,
this would contradict \cref{thm:f-i-f-i}.

Now, consider a position $x \in (\Theta_{i,k-2}  \oplus f_{i-k})$. 
Assume, by contradiction, that $x \in \Theta_{i,k-1}$, 
then, we have $F_{i-(k-2)} = F_{i-k} \ F_{i-(k-1)}$,
which contradicts 
$F_{i-(k-2)} = F_{i-(k-1)} \ F_{i-k} \neq F_{i-k} \ F_{i-(k-1)}$.
This is analogous to $F_{i-1} \ F_{i-2} \neq F_{i-2} \ F_{i-1}$, 
which follows from the ``near-commutative property'' of Fibonacci words~\cite{DBLP:journals/siamcomp/KnuthMP77}.
Thus, $\Theta_{i,k-1} \cap (\Theta_{i,k-2}  \oplus f_{i-k}) = \emptyset$.
Next, by the induction hypothesis, 
$\max(\Theta_{i,k-1} \cup \Theta_{i,k-2}) = f_i - f_{i-(k-2)} + 1$.
Note that $(f_i - f_{i-(k-2)} + 1) + f_{i-k} < f_i - f_{i-k} + 1$.
Therefore, $f_i - f_{i-k} + 1 \notin \Theta_{i,k-1} \cup (\Theta_{i,k-2}  \oplus f_{i-k})$
and $\max(\Theta_{i, k}) = f_i - f_{i-k} + 1$.

The proof for odd $k$ is very similar to even $k$ with the difference being that we do not need to consider  $f_i - f_{i-k} + 1 $ for odd $k$.
\end{proof}

In the following result, 
the case where $0 \leq j \leq  i-4$ 
has been addressed in \cite{journal/tcs/2006/rytter},
while the case where $i-3 \leq j \leq i-1$ is straightforward.
Our characterization in \cref{thm:fib-occ-inductive}
can offer an alternative simpler proof for this result.

\begin{corollary}\label{thm:fib-occ-num}
Consider $F_i$ and $0 \leq j \leq i-1$.
Let $\theta_{i,j}$ denote the number of occurrences of $F_{i-j}$ in $F_i$,
and define $f_{-1} := 1$ and $f_0 := 0$ for convenience.
Then,
\[
\theta_{i,j} = 
\begin{cases}
\begin{aligned}
& f_{j+2} - (j \bmod 2) & \text{~if~} & \quad 0 \leq j \leq  i-4; \\
& f_{j+1}               & \text{~if~} & \quad i-3 \leq j \leq i-2; \\
& f_{j-1}               & \text{~if~} & \quad j = i-1.
\end{aligned}
\end{cases}
\]
\end{corollary}

\section{Occurrences of Thue-Morse Words of Smaller Order}\label{sec:occ-tm-smaller-ord}

\begin{figure}[t]
\centering
\includegraphics[width=0.9\linewidth]{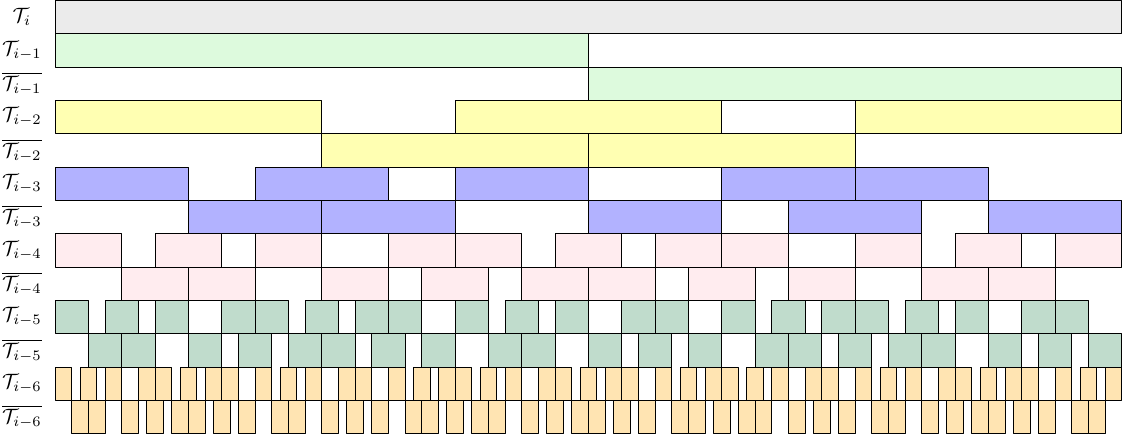}
\caption{
An illustration of the occurrences of $\tm_{i-j}$ and $\flip{\tm_{i-j}}$ 
in $\tm_{i}$ for $1 \leq j \leq 6$.
}
\label{fig:tm-occ}
\end{figure}

We study the occurrences of $\tm_{i-j}$ and $\flip{\tm_{i-j}}$ in each $\tm_i$ for appropriate $i$ and $j$
(the occurrences are shown in \cref{fig:tm-occ} for $1 \leq j \leq 6$).
These results will help us identify the net occurrences in 
each $\tm_i$ in \cref{sec:fib-net-occ} and  may also be of independent interest.
We now present the main result of the section,
illustrated in \cref{fig:tm-a-occ-cases}.

\begin{figure}[t]
\centering
\includegraphics[width=0.9\linewidth]{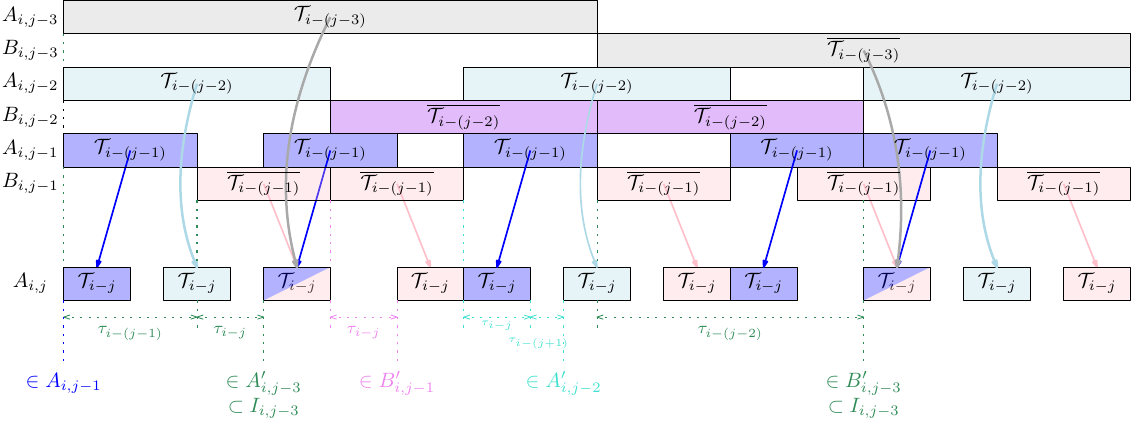}
\caption{
An illustration of \cref{thm:tm-occ}.
Each dark blue, pink, and light blue occurrence of $\tm_{i-j}$ is introduced by
an occurrence of $\tm_{i-(j-1)}$, $\flip{\tm_{i-(j-1)}}$, and $\tm_{i-(j-2)}$ respectively.
Each occurrence of $\tm_{i-j}$ that is both dark blue and pink 
indicates that it is introduced by both an occurrence of $\tm_{i-(j-1)}$ and an occurrence of $\flip{\tm_{i-(j-1)}}$.
}
\label{fig:tm-a-occ-cases}
\end{figure}

\begin{theorem}\label{thm:tm-occ}
For each $i \geq 2$ and $0 \leq j \leq i-1$,
let $A_{i,j}$ and $B_{i,j}$
denote the set of the starting positions of the occurrences 
of $\tm_{i-j}$ and $\flip{\tm_{i-j}}$
in $\tm_i$, respectively.
Then, $A_{i,0} = A_{i,1} = \{1\}$.
For each $j \geq 2$,
we define 
\begin{align*}
B_{i,j-1}' &:= B_{i,j-1} \oplus \ltm_{i-j},  \qquad  
A_{i,j-2}' := A_{i,j-2} \oplus \left(\ltm_{i-j} + \ltm_{i-(j+1)}\right), \\
I_{i,j-3} &:= 
\begin{cases}
\emptyset, & j = 2, \\
A_{i, j-3}' \cup B_{i, j-3}', &  j \geq 3,  \text{~where~}
\end{cases} \\
A_{i, j-3}' &:= A_{i,j-3} \oplus (\ltm_{i-(j-1)} + \ltm_{i-j}),  \qquad 
B_{i, j-3}' := B_{i,j-3} \oplus \ltm_{i-(j-2)}.
\end{align*} 
Then,
$A_{i,j} =  A_{i,j-1}  \cup  B_{i,j-1}' \cup  A_{i,j-2}'$
with 
\[
A_{i,j-1}   \cap B_{i,j-1}' = I_{i,j-3},  \quad
A_{i,j-1}   \cap A_{i,j-2}' = \emptyset,  \quad \text{and} \quad
B_{i,j-1}'   \cap A_{i,j-2}' = \emptyset.
\]
\end{theorem}

\begin{proof}
We proceed by induction on $j$.

\textbf{Base cases.}
When $j = 0$, 
the claim holds trivially.
When $j = 1$, 
note that $\tm_{i-1}$ only occurs at position 1 because
it cannot occur at position $\ltm_{i-1}+1$ (where $\flip{\tm_{i-1}}$ occurs),
and any other occurrences of $\tm_{i-1}$ would 
overlap with its occurrence at position 1, 
contradicting \cref{thm:tm-overlap-free}.
When $j = 2$, observe that 
$A_{i,2-1} = \{1\}$, 
$B_{i,2-1}' = \{ \ltm_{i-1} + \ltm_{i-2} + 1 \}$ and
$A_{i,2-2}' = \{ \ltm_{i-2} + \ltm_{i-3} + 1 \}$
are  mutually disjoint.
Further, there are no occurrences of $\tm_{i-2}$
outside of $A_{i,2} = A_{i,1} \cup B_{i,1}' \cup A_{i,0}'$ 
because any such occurrences would contradict \cref{thm:tm-overlap-free}.

\textbf{Inductive step.}
For each $3 \leq k \leq i-1$, assume the claim holds for $j=k-3$, $k-2$, and $k-1$ and 
we now prove the claim for $j=k$. 
Define $V_{i, k} := A_{i,k-1}  \cup  B_{i,k-1}' \cup  A_{i,k-2}'$.
We prove $A_{i, k} = V_{i, k}$ by
showing $A_{i, k} \subset V_{i, k}$ and $V_{i, k} \subset A_{i, k}$.

To prove  $V_{i, k} \subset A_{i, k}$,
we will show that each set in the union 
defining $V_{i, k}$ is contained in $A_{i, k}$.
Clearly, $A_{i,k-1} \subset A_{i, k}$ 
because $\tm_{i-j}$ is a prefix of
$\tm_{i-(j-1)} = \tm_{i-j} \ \flip{\tm_{i-j}}$.
Similarly, $B_{i,k-1}' \subset A_{i, k}$  
because $\tm_{i-j}$ is a suffix of
$\flip{\tm_{i-(j-1)}} =  \flip{\tm_{i-j}} \ \tm_{i-j}$.
Lastly, $A_{i,k-2}' \subset A_{i, k}$
because $\tm_{i-j}$ occurs at position $\ltm_{i-k} + \ltm_{i-(k+1)}$ of 
\begin{equation}\label{eq:t-i-k-2}
\tm_{i-(k-2)} 
= \tm_{i-k} \ \flip{\tm_{i-(k+1)}} \ \tm_{i-k} \ \tm_{i-(k+1)} \ \tm_{i-k}
\end{equation}

Next, we prove $A_{i, k} \subset V_{i, k}$
by showing that each occurrence of $\tm_{i-k}$ is in $V_{i, k}$.
By \cref{thm:tm-basic-fac}, there is a factorization of $\tm_i$ 
where each factor is either $\tm_{i-(k-1)}$ or $\flip{\tm_{i-(k-1)}}$.
We thus consider the following four cases.
\begin{enumerate}[{Case} 1]
\item\label{tm-occ-case-1}
When $\tm_{i-k}$ occurs within $\tm_{i-(k-1)}  \ \tm_{i-(k-1)} 
= \tm_{i-k} \ \flip{\tm_{i-k}} \ \tm_{i-k} \ \flip{\tm_{i-k}}$,
by the overlap-free property (\cref{thm:tm-overlap-free}),
positions 1 and $\ltm_{i-(k-1)} + 1$ are the only two occurrences
of $\tm_{i-k}$.
They are both contained in $A_{i,k-1}$,
while the latter is also in $B_{i,k-1}'$.
(The overlap-free property will be used similarly in the remaining three cases.)
\item\label{tm-occ-case-2}
When $\tm_{i-k}$ occurs within 
$\flip{\tm_{i-(k-1)}} \ \flip{\tm_{i-(k-1)}} 
= \flip{\tm_{i-k}} \ \tm_{i-k} \ \flip{\tm_{i-k}} \ \tm_{i-k}$,
positions $\ltm_{i-k} + 1$ and $\ltm_{i-(k-1)} + \ltm_{i-k} + 1$
are the only two occurrences of $\tm_{i-k}$.
They are both contained in $B_{i,k-1}'$,
while the former is also in $A_{i,k-1}$.
\item
When $\tm_{i-k}$ occurs within 
$\tm_{i-(k-1)} \ \flip{\tm_{i-(k-1)}} 
= \tm_{i-(k-2)}
$,
by \cref{eq:t-i-k-2},
positions 1, $\ltm_{i-k} + \ltm_{i-(k+1)} + 1$ and $\ltm_{i-(k-1)} + \ltm_{i-k} + 1$
are the only occurrences of $\tm_{i-k}$:
the first and third are both contained in $A_{i,k-1}$,
the second is in $A_{i,k-2}'$,
and the third is also in $B_{i,k-1}'$.
\item
When $\tm_{i-k}$ occurs within 
$\flip{\tm_{i-(k-1)}} \ \tm_{i-(k-1)}
= \flip{\tm_{i-k}} \ \tm_{i-k} \ \tm_{i-k} \ \flip{\tm_{i-k}}$,
position $\ltm_{i-k}$ and $\ltm_{i-(k-1)}$
are the only two occurrences of $\tm_{i-k}$.
The former is contained in $B_{i,k-1}'$,
while the latter is in $A_{i,k-1}$.
\end{enumerate}
After examining the above four cases, 
we conclude that $A_{i, k} \subset V_{i, k}$,
and thus $A_{i, k} = V_{i, k}$. 
Next we will prove $A_{i,k-1}   \cap B_{i,k-1}' = I_{i,k-3}$ 
by showing $A_{i,k-1}   \cap B_{i,k-1}' \subset I_{i,k-3}$ and 
$I_{i,k-3} \subset A_{i,k-1}   \cap B_{i,k-1}'$. 
Recall that $I_{i,k-3} := A_{i, k-3}' \cup B_{i, k-3}'$.

First, we prove $A_{i,k-1}   \cap B_{i,k-1}' \subset I_{i,k-3}$
by establishing that if an occurrence of $\tm_{i-k}$
is in $A_{i,k-1}   \cap B_{i,k-1}'$,
then this occurrence is in $I_{i,k-3}$.
First observe that 
in  Cases~\ref{tm-occ-case-1}--\ref{tm-occ-case-2},
some occurrences of $\tm_{i-k}$  are contained in both
$A_{i,k-1}$ and $B_{i,k-1}'$.
By \cref{thm:tm-overlap-free},
$\flip{\tm_{i-(k-3)}}$ and $\tm_{i-(k-3)}$ do not overlap in $\tm_i$,
it follow that, 
for each occurrence of $\flip{\tm_{i-(k-3)}}$ in $\tm_i$,
there is only one occurrence of $\tm_{i-k}$ contained in $B_{i,k-3}'$.
Similarly, 
for each occurrence of $\tm_{i-(k-3)}$ in $\tm_i$,
there is only one occurrence of $\tm_{i-k}$ contained in $A_{i,k-3}'$.
Now, consider the factorizations:
\begin{align*}
\flip{\tm_{i-(k-3)}} &= 
\flip{\tm_{i-(k-1)}} \ \tm_{i-(k-1)} \ \tm_{i-(k-1)} \ \flip{\tm_{i-(k-1)}}, 
\text{~and~} \\ 
\tm_{i-(k-3)} &= 
\tm_{i-(k-1)} \ \flip{\tm_{i-(k-1)}} \ \flip{\tm_{i-(k-1)}} \ \tm_{i-(k-1)}.
\end{align*}
Observe that the set of occurrences 
of $\tm_{i-k}$ in Case~\ref{tm-occ-case-1} 
is a subset of $B_{i,k-3}'$  
since $\tm_{i-(k-1)} \ \tm_{i-(k-1)}$ occurs at position $\ltm_{i-(k-1)}+1$ in
$\flip{\tm_{i-(k-3)}}$.
Similarly, the set 
of occurrences  of $\tm_{i-k}$ in Case~\ref{tm-occ-case-2} 
is a subset of $A_{i,k-3}'$ since $\flip{\tm_{i-(k-1)}} \ \flip{\tm_{i-(k-1)}}$ 
occurs at position $\ltm_{i-(k-1)}+1$ in 
$\tm_{i-(k-3)}$.

Next, we prove  $I_{i,k-3} \subset A_{i,k-1}   \cap B_{i,k-1}'$ by contraposition. 
Specifically, instead of directly showing that
``if  an occurrence of $\tm_{i-k}$ is in $I_{i,k-3}$,
then this occurrence is in $A_{i,k-1} \cap B_{i,k-1}'$'',
we prove the equivalent contrapositive:
``if an occurrence of $\tm_{i-k}$ is not in $A_{i,k-1} \cap B_{i,k-1}'$,
then this occurrence is not in $I_{i,k-3}$''.
First observe that $\tm_{i-k}$ occurs at position 1 in 
$\tm_{i-(k-1)} = \tm_{i-k} \ \flip{\tm_{i-k}}$
and occurs at position $\ltm_{i-k}+1$ in 
$\flip{\tm_{i-(k-1)}} = \flip{\tm_{i-k}} \ \tm_{i-k}$.
Next, occurrences of $\tm_{i-k} \ \flip{\tm_{i-k}}$
and $\tm_{i-k} \ \flip{\tm_{i-k}}$ overlap in $\tm_i$  
to form occurrences of
$\flip{\tm_{i-k}} \ \tm_{i-k} \ \flip{\tm_{i-k}}$ 
(see Cases~\ref{tm-occ-case-1}--\ref{tm-occ-case-2}).
Hence,
if an occurrence of $\tm_{i-k}$ is not in $A_{i,k-1} \cap B_{i,k-1}'$,  
then this occurrence is not at position $\ltm_{i-k}+1$ in 
$\flip{\tm_{i-k}} \ \tm_{i-k} \ \flip{\tm_{i-k}}$. 
Next, consider the factorizations:
\begin{align*}
\flip{\tm_{i-(k-3)}} &= 
\flip{\tm_{i-(k-1)}} \ \tm_{i-k} \ \flip{\tm_{i-k}} \ \tm_{i-k} \ \flip{\tm_{i-k}} \ \flip{\tm_{i-(k-1)}}, \text{~and~} \\
\tm_{i-(k-3)} &= 
\tm_{i-(k-1)} \ \flip{\tm_{i-k}} \  \tm_{i-k} \ \flip{\tm_{i-k}} \ \tm_{i-k} \ \tm_{i-(k-1)}.
\end{align*}
Since $\flip{\tm_{i-(k-3)}}$ and $\tm_{i-(k-3)}$
do not overlap in $\tm_i$,
we know that 
$\flip{\tm_{i-k}} \ \ \tm_{i-k} \ \flip{\tm_{i-k}}$ only occurs 
at position $\ltm_{i-(k-1)}+\ltm_{i-k}+1$ in $\flip{\tm_{i-(k-3)}}$ and 
at position $\ltm_{i-(k-1)}+1$ in $\tm_{i-(k-3)}$.
Thus, if an occurrence of $\tm_{i-k}$ is not in $A_{i,k-1} \cap B_{i,k-1}'$,
then this occurrence is not in $I_{i,k-3}$.
Therefore, we conclude that $I_{i,k-3} \subset A_{i,k-1}   \cap B_{i,k-1}'$.
\end{proof}

The analogous characterization of $B_{i, j}$ is 
presented as follows,
which
can be proven in a way similar to the proof of \cref{thm:tm-occ}.

\begin{corollary}\label{thm:tm-b-occ}
For each $i \geq 2$ and $0 \leq j \leq i-1$,
we have 
$B_{i,0} = \emptyset$, $B_{i,1} = \{\ltm_{i-1} + 1 \}$.
For each $j \geq 2$, we define 
\begin{align*}
A_{i,j-1}'' &:= A_{i,j-1} \oplus \ltm_{i-j}, \qquad  
B_{i,j-2}'' := B_{i,j-2} \oplus (\ltm_{i-j} + \ltm_{i-(j+1)}), \\
I_{i,j-3}' &:= 
\begin{cases}
\emptyset, & j = 2, \\
A_{i, j-3}'' \cup B_{i, j-3}'', &  j \geq 3,  \text{~where~}
\end{cases} \\
B_{i, j-3}'' &:=  B_{i,j-3} \oplus (\ltm_{i-(j-1)} + \ltm_{i-j}) , \qquad
A_{i, j-3}'' := A_{i,j-3} \oplus \ltm_{i-(j-2)} .
\end{align*} 
Then,
$B_{i,j} =  B_{i,j-1}  \cup  A_{i,j-1}'' \cup  B_{i,j-2}''$
with 
\[
B_{i,j-1}   \cap A_{i,j-1}'' = I'_{i,j-3},  \quad
B_{i,j-1}   \cap B_{i,j-2}'' = \emptyset,  \quad \text{and} \quad
A_{i,j-1}''   \cap B_{i,j-2}'' = \emptyset.
\lipicsEnd 
\]
\end{corollary}

We next use \cref{thm:tm-occ} and \cref{thm:tm-b-occ} to count the number of
occurrences of $\tm_{i-j}$ and $\flip{\tm_{i-j}}$
in $\tm_i$.

\begin{corollary}\label{thm:tm-occ-num}
For $i \geq 2$,
consider $\tm_i$ and $0 \leq j \leq i-1$.
Let $a_{j}$ and $b_{j}$
denote the number of occurrences of $\tm_{i-j}$ and $\flip{\tm_{i-j}}$
in $\tm_i$, respectively.
Then, 
\begin{itemize}
\item 
$a_{0} = a_{1} = 1$ and
$a_{j} = a_{j-1} + 2 \cdot a_{j-2}$
for each $j \geq 2$;
\item 
$b_{0} = 0$ and 
$b_{j} = b_{j-1} + a_{j-1}$
for each $j \geq 1$.
\end{itemize}
\end{corollary}

\begin{proof}
We proceed by induction on $j$.
For $1 \leq j \leq 2$, 
the claim holds trivially.
For $j \geq 3$, we have $|I_{i,j-3}| 
= |A_{i,j-3}| + |B_{i,j-3}|  = a_{j-3} + b_{j-3} = b_{j-2}$
and 
$
  |A_{i,j}| 
= |A_{i,j-1}| + |B_{i,j-1}'| + |A_{i,j-2}'| - |I_{i,j-3}| 
= a_{j-1} + b_{j-1} + a_{j-2} - b_{j-2}
= a_{j-1} + (b_{j-1} - b_{j-2}) + a_{j-2}
= a_{j-1} + 2 \cdot a_{j-2}
$.
by induction hypothesis.
We can prove $b_j$ similarly with \cref{thm:tm-b-occ}.
\end{proof}

\begin{remark}
For each $i \geq 2$,
$a_{j}$ is the $(j+1)^\text{th}$ Jacobsthal Number: Sequence A001045 of the On-Line Encyclopedia of Integer Sequences (\url{https://oeis.org/A001045}). 
\end{remark}

With the characterization 
of the occurrences of $\tm_{i-j}$ and $\flip{\tm_{i-j}}$ in $\tm_i$,
a natural next step is to investigate the structure of the strings that surround 
$\tm_{i-j}$ and $\flip{\tm_{i-j}}$,
which correspond to the blank areas in each row in \cref{fig:tm-occ}.
In \cref{sec:tm-fac},
we explore two smallest factorizations of 
$\tm_i$, each containing all occurrences of $\tm_{i-j}$ and $\flip{\tm_{i-j}}$, respectively.
The remaining factors in these factorizations represent the surrounding strings.

\section{Net Occurrences in Fibonacci Words}\label{sec:fib-net-occ}
In this section,
we prove that there are only three net occurrences in each $F_i$,
using the results on the occurrences of Fibonacci words of smaller order from \cref{sec:occ-fib-smaller-ord},
the notion of ONOC from \cref{sec:onoc},
and new properties that we will develop in this section.
We begin by reviewing the following results. 
Some of the proofs in this section are presented in \cref{appx:fib-net-occ-proofs}.

\begin{lemma}[\cite{conf/cpm/2024/guo}]\label{thm:q-i-delta}
For $i \geq 7$,
let $Q_i := F_{i-5}  F_{i-6} \cdots F_3  F_2$,
$\Delta(0) := \texttt{ba}$, 
$\Delta(1) := \texttt{ab}$.
Then,
\begin{align}
F_{i-4} \ F_{i-5} &= Q_i \ \Delta(1 - (i \bmod 2)) \text{~and~} \\
F_{i-5} \ F_{i-4} &= Q_i \ \Delta(i \bmod 2).
\end{align}
\end{lemma}

\begin{lemma}[\cite{conf/cpm/2024/guo}]\label{thm:net-occ-fib}
For each $i \geq 7$,
the following are net occurrences in $F_i$:
\begin{itemize}
\item one occurrence of $F_{i-2}$ at position $f_{i-1}+1$;
\item two occurrences of $F_{i-2} \ Q_i$ at positions $1$ and $f_{i-2}+1$.
\end{itemize}
Meanwhile, the two occurrences of $F_{i-2}$ at positions $1$ and $f_{i-2}+1$  are \emph{not} net occurrences.
\end{lemma}

With \cref{thm:f-i-2}, we know that these three are the only occurrences of $F_{i-2}$ in $F_i$.
Similarly, we now strengthen \cref{thm:net-occ-fib}
by showing the following result.

\begin{restatable}{lemma}{fitwoqiocc}\label{thm:f-i-2-q-i-occ}
For each $i \geq 7$,
$F_{i-2} \ Q_i$ only occurs at positions 1 and $f_{i-2}+1$ in $F_i$.
\end{restatable}

By combining \cref{thm:f-i-2}, \cref{thm:net-occ-fib}, 
and \cref{thm:f-i-2-q-i-occ},
we conclude that $F_{i-2}$ has only one net occurrence
and $F_{i-2} \ Q_i$ only has two net occurrences.

\begin{lemma}
For each $i \geq 7$,
the net occurrences identified in \cref{thm:net-occ-fib} are the only net occurrences of $F_{i-2}$ and $F_{i-2} \ Q_i$ in $F_i$.
\end{lemma}

\begin{figure}[t]
\centering
\includegraphics[width=0.5\linewidth]{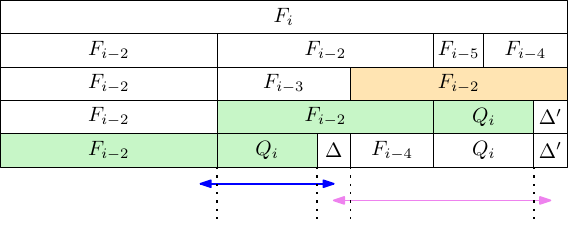}
\caption{
An illustration of several factorizations of $F_i$ 
from \cref{thm:f-i-fac} and \cref{thm:q-i-delta}
where $\Delta := \Delta(1 - (i \bmod 2))$ and 
$\Delta' := \Delta(i \bmod 2)$.
Net occurrences of $F_{i-2}$ and $F_{i-2} \ Q_i$
are in yellow and green, respectively.
Super-occurrences of the two BNSOs 
are shown as arrows. 
}
\label{fig:fib-non-net-occ-cases}
\end{figure}

It remains to show that there are no additional net occurrences in each $F_i$.
To achieve this, we use the results from \cref{sec:onoc}.
First, observe that 
the three net occurrences in \cref{thm:net-occ-fib} form an ONOC of $F_i$.
The two BNSOs of this ONOC 
correspond to 
an occurrence of $Q_i$ and 
an occurrence of $F_{i-4} \ Q_i$, respectively.
See \cref{fig:fib-non-net-occ-cases} for an illustration.
Next, we aim to show that no super-occurrences of these two occurrences 
can be a net occurrence.
To establish this,
we analyze the super-occurrences of 
the occurrences of $F_{i-3}$ in \cref{thm:f-i-3-sub-occ}.
This result  covers the examination of the super-occurrences 
of the occurrence of  $F_{i-4} \ Q_{i}$,
since $F_{i-3}$ is a prefix of $F_{i-4} \ Q_{i} = F_{i-3} \ Q_{i-1}$.
Furthermore,
\cref{thm:f-i-3-sub-occ} helps examining super-occurrences of
the occurrence of $Q_i$ 
in \cref{thm:q-i-sub-occ}.

To prove these two lemmas, we introduce some properties of $F_{i-3}$ and $Q_{i}$,
which are proved in \cref{appx:fib-net-occ-proofs}.

\begin{restatable}{lemma}{qlen}\label{thm:q-len}
$|Q_i| = f_{i-3} - 2$.
\end{restatable}

\begin{restatable}{lemma}{fitwoqi}\label{thm:f-i-2-q-i}
For a substring $S$ of $F_i$,
if $F_{i-2} \ Q_i$ is a proper substring of $S$,
then $S$ is unique.
\end{restatable}

\begin{restatable}{lemma}{qione}\label{thm:q-i-1}
$F_{i-3}$ is always followed by 
 $Q_{i-1}$ in $F_i$.
\end{restatable}

\begin{restatable}{lemma}{fitsf}\label{thm:f-i-365}
$F_{i-3} \ F_{i-6} \ F_{i-5}$ and its length-$(f_{i-2}-1)$ prefix are both unique in $F_i$.
\end{restatable}

\begin{restatable}{lemma}{fithreerightext}\label{thm:f-i-3-right-ext}
The length-$(f_{i-3}-1)$ prefix of $F_{i-3}$
is always followed by $F_{i-3}[f_{i-3}]$ in $F_i$.
\end{restatable}

Now, we introduce the two crucial lemmas motivated earlier.

\begin{figure}
\centering

\begin{subfigure}{0.45\linewidth}
\centering
\includegraphics[width=\linewidth]{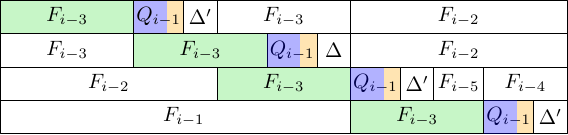}
\caption{Case (\ref{fib-case-1})}
\end{subfigure}
\hfill
\begin{subfigure}{0.45\linewidth}
\centering
\includegraphics[width=\linewidth]{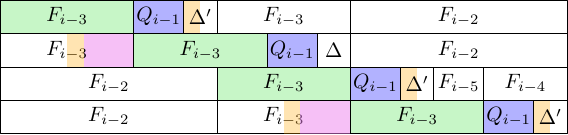}
\caption{Case (\ref{fib-case-2})}
\end{subfigure}

\medskip

\begin{subfigure}{0.45\linewidth}
\centering
\includegraphics[width=\linewidth]{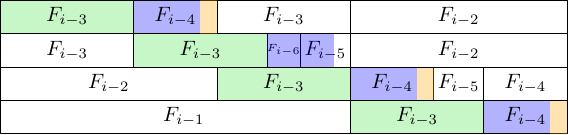}
\caption{Case (\ref{fib-case-3})}
\end{subfigure}
\hfill
\begin{subfigure}{0.45\linewidth}
\centering
\includegraphics[width=\linewidth]{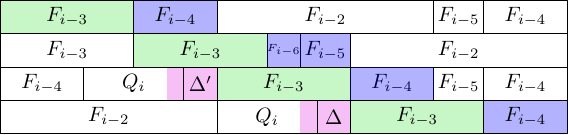}
\caption{Case (\ref{fib-case-4})}
\end{subfigure}

\medskip

\begin{subfigure}{0.45\linewidth}
\centering
\includegraphics[width=\linewidth]{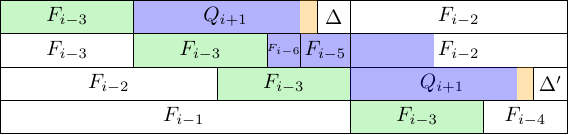}
\caption{Case (\ref{fib-case-5})}
\end{subfigure}
\hfill
\begin{subfigure}{0.45\linewidth}
\centering
\includegraphics[width=\linewidth]{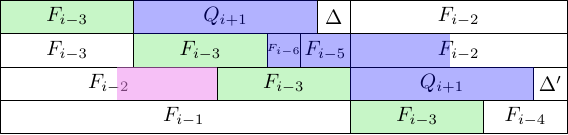}
\caption{Case (\ref{fib-case-6})}
\end{subfigure}

\medskip

\begin{subfigure}{0.45\linewidth}
\centering
\includegraphics[width=\linewidth]{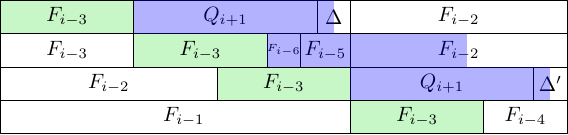}
\caption{Case (\ref{fib-case-7})}
\end{subfigure}
\caption{
Illustration of the proof of \cref{thm:f-i-3-sub-occ}.
In each case, four factorizations of $F_i$ are shown,
each focusing on one occurrence of $F_{i-3}$, highlighted in green.
For $S = X \ F_{i-3} \ Y$,
each discussed $X$ and $Y$ is shown in pink and blue, respectively.
Each discussed left or right extension character of $S$ is shown in yellow.
Recall that $\Delta := \Delta(1 - (i \bmod 2))$ and 
$\Delta' := \Delta(i \bmod 2)$.
}
\label{fig:proof-7-cases}
\end{figure}

\begin{lemma}\label{thm:f-i-3-sub-occ}
Consider an occurrence $(s, e)$ in $F_i$
and let $S := F_i[s \ldots e]$.
If  $(s, e)$ is a super-occurrence of
an occurrence of $F_{i-3}$,
and $S$ is neither $F_{i-2}$ nor $F_{i-2} \ Q_i$,
then $(s, e)$ is not a net occurrence.
\end{lemma}

\begin{proof}
The proof is illustrated in \cref{fig:proof-7-cases}.
Consider strings $X$ and $Y$ such that $S = X  \ F_{i-3} \ Y$ 
and  $X \ Y \neq \epsilon$.
We examine the following cases depending on $|Y|$.
Note that $|Q_{i-1}| = f_{i-4}-2$ and $|Q_{i+1}| = f_{i-2}-2$
from \cref{thm:q-len}.

\begin{enumerate}[(1)] 
\item\label{fib-case-1}
$|Y| < |Q_{i-1}|$.
Using \cref{thm:q-i-1},
note that $Y$ is a prefix of $Q_{i-1}$ in this case. 
This means the right extension character of $S$ is always $Q_{i-1}[|Y|+1]$.
Thus, no occurrence of $S$ is a net occurrence.

\item\label{fib-case-2}
$|Y| = |Q_{i-1}|$.
Using \cref{thm:q-i-1},
$F_{i-3} \ Y = F_{i-3} \ Q_{i-1}$ always holds in this case.
Next, if $F_{i-3} \ Y$ occurs at position $1$, $f_{i-2}+1$, 
or $f_{i-1}+1$, then the right extension character 
is always $\Delta(i \bmod 2)[1]$.
On the other hand, if $F_{i-3} \ Y$ occurs at position $f_{i-3}+1$,
we examine the left extension character of 
$S = XF_{i-3}Y$.
Notice that $|X| \leq f_{i-3}$ always holds in this case
(and $S$ becomes a prefix of $F_i$ when $|X| = f_{i-3}$).
Now, observe that if $F_{i-3} \ Y$ occurs at position $f_{i-3}+1$ or $f_{i-1}+1$,
the left extension character of $S$ is always 
$F_{i-3}[f_{i-3} - |X| - 1]$.
Thus, this occurrence of $S$ is also not a net occurrence.

\item\label{fib-case-3}
$|Q_{i-1}| < |Y| < f_{i-4}$.
If $F_{i-3} \ Y$ occurs at positions $1$, $f_{i-2}+1$, or $f_{i-1}+1$,
observe that occurrences of $F_{i-3}$ 
at these three positions are always followed by $F_{i-4}$.
Thus, $Y$ is a prefix of $F_{i-4}$
and the right extension character of $S$ is always 
$F_{i-4}[|Y|+1]$.
So these three occurrences of $S$ are not net occurrences.
If $F_{i-3} \ Y$ occurs at position $f_{i-3}+1$,
since $|Y| =|Q_{i-1}|+1 = f_{i-4} - 1$,
we have $|F_{i-3} \ Y| = f_{i-3} + f_{i-4} - 1 = f_{i-2} - 1$.
By \cref{thm:f-i-365}, $F_{i-3} \ Y$ is unique,
which means $S$ is unique.

\item\label{fib-case-4}
$|Y| = f_{i-4}$.
If $F_{i-3}  Y$ occurs at position 1,
then $X$ is empty and $S = F_{i-3} \ F_{i-4} = F_{i-2}$. 
If $F_{i-3}  Y$ occurs at positions $f_{i-3}+1$,
then $F_{i-3}  Y = F_{i-3} \ F_{i-6} \ F_{i-5}$ is unique (\cref{thm:f-i-365})
so $S$ is also unique.
If $F_{i-3} Y$ occurs at position $f_{i-2}+1$,
then $X$ ends with $\Delta(i \bmod 2)$. 
If $F_{i-3} Y$ occurs at position $f_{i-1}+1$,
then $X$ ends with $\Delta(1 - (i \bmod 2))$.
Now, since $\Delta(i \bmod 2) \ F_{i-2}$
and $\Delta(1 - (i \bmod 2))\ F_{i-2}$ are both unique by \cref{thm:f-i-2},
$S$ is also unique if $F_{i-3} Y$ occurs at these two positions.

\item\label{fib-case-5}
$f_{i-4} < |Y| < |Q_{i+1}|$.
First observe that 
the occurrences of $F_{i-3}$ at positions 
$1$ and $f_{i-2}+1$ are both followed by $F_{i-4}\ Q_i = Q_{i+1}$.
Thus, if $F_{i-3}Y$ occurs at these two positions, 
then $Y$ is a prefix of $Q_{i+1}$
and the right extension character of $S$ is always $Q_{i+1}[|Y|+1]$.
So these two occurrences of $S$ are not net occurrences.
Next, if $F_{i-3} \ Y$ occurs at position $f_{i-3}+1$,
then $F_{i-3} \ Y$ is unique
because $F_{i-3} \ F_{i-6} \ F_{i-5}$ is a prefix of $F_{i-3} Y$
and the former is unique by \cref{thm:f-i-365}.
Thus, $S$ is unique.
Finally note that, $F_{i-3} Y$ 
cannot occur at position $f_{i-1}+1$ because
$|Y| > |F_{i-4}|$ and  $F_{i-3} \ F_{i-4}$ is a suffix of $F_i$.

\item\label{fib-case-6}
$|Y| = |Q_{i+1}|$.
Similar to the previous case,
if $F_{i-3}Y$ occurs at position $f_{i-3}+1$,
then $F_{i-3} Y$ is unique,
and $F_{i-3}Y$ cannot occur at position $f_{i-1}+1$.
If $F_{i-3}Y$ occurs at position 1,
then $X$ is empty and 
$S = F_{i-3}Y = F_{i-3} \  Q_{i+1} = F_{i-2} \ Q_i$.
If $F_{i-3}Y$ occurs at position $f_{i-2}+1$,
then $S = X \ F_{i-2} \ Q_i$,
which is unique by \cref{thm:f-i-2-q-i}.

\item\label{fib-case-7}
$|Y| > |Q_{i+1}|$.
Similar to the previous two cases,
if $F_{i-3}Y$ occurs at position $f_{i-3}+1$,
then $F_{i-3} Y$ is unique,
and $F_{i-3}Y$ cannot occur at position $f_{i-1}+1$.
If $F_{i-3}Y$ occurs at position $f_{i-2}+1$,
then $F_{i-3}Y$ is a prefix of 
$F_{i-3} \ F_{i-2} = 
F_{i-2} \ F_{i-5} \ F_{i-4} =
F_{i-2} \ Q_i \ \Delta(i \bmod 2)$,
which is unique by \cref{thm:f-i-2-q-i}.
Thus, $S$ is unique. 
Finally, 
since $|Y| > |Q_{i+1}| > f_{i-2}-1$, 
if $F_{i-3}Y$ occurs at positions $1$,
then $Y$ begins with 
the length-$(f_{i-2}-1)$ prefix of
$F_{i-3} \ F_{i-6} \ F_{i-5} 
$,
which is unique by \cref{thm:f-i-365}.
Thus, $S$ is also unique. 
\qedhere
\end{enumerate}
\end{proof}

\begin{lemma}\label{thm:q-i-sub-occ}
Consider an occurrence $(s, e)$ in $F_i$.
If  $(s, e)$ is a proper super-occurrence of
the occurrence of $Q_i$ at position $f_{i-2}+1$,
then $(s, e)$ is not a net occurrence.
\end{lemma}

\begin{proof}
Consider strings $X$ and $Y$ such that $S = X  Q_i Y$
and  $X \ Y \neq \epsilon$.
When $|Y| \geq 2$, notice that $F_{i-3}$ occurs in $S$.
By \cref{thm:f-i-3-sub-occ}, $(s, e)$ is not a net occurrence.
Now, we consider the case when $|Y| = 1$.
Note that $Q_i Y$ is precisely the length-($f_{i-3}-1$) prefix of 
$F_{i-3}$.
Thus, by \cref{thm:f-i-3-right-ext},
$Q_i Y$ is always followed by the same right extension character, 
$F_{i-3}[f_{i-3}]$, which means $(s, e)$ is not a net occurrence.
\end{proof}

Finally, the main result follows from \cref{thm:net-occ-fib}, \cref{thm:super-occ-bnso},  \cref{thm:f-i-3-sub-occ} and \cref{thm:q-i-sub-occ}.

\begin{theorem}\label{thm:fib-only-net-occ}
The three net occurrences identified in \cref{thm:net-occ-fib}
are the only ones in $F_i$.
\end{theorem}

\section{Net Occurrences in Thue-Morse Words}\label{sec:tm-net-occ}

In this section,
we prove the only nine net occurrences in each $\tm_i$
using the results on the occurrences of Thue-Morse words of smaller order from \cref{sec:occ-tm-smaller-ord},
the notion of ONOC from \cref{sec:onoc},
and new results that we will introduce in this section.
We will first show that each occurrence of each string 
in $\mathcal{P}_i$ (defined below) 
is a net occurrence in $\tm_i$, then show that they are the only ones.

\begin{definition}\label{def:tm-pos-nf-str}
For each $i \geq 5$,
define $\mathcal{P}_i := \{ 
\tm_{i-2}, \
\flip{\tm_{i-2}}, \
\tm_{i-4} \ \flip{\tm_{i-3}}, \
\flip{\tm_{i-4}} \ \tm_{i-3} \}$..
\end{definition}

We next show several factorizations of $\tm_i$,
proved in \cref{appx:tm-net-occ-proofs}.

\begin{restatable}{lemma}{tmifac}\label{thm:tm-i-fac}
For each $i \geq 5$:
\begin{align}
\tm_{i} &= \tm_{i-2} \ \flip{\tm_{i-2}} \ \flip{\tm_{i-2}} \ \tm_{i-2} 
\\
\tm_{i} &= \tm_{i-2} \ \flip{\tm_{i-3}} \ \tm_{i-2} \ \tm_{i-3} \ \tm_{i-2} 
\\
\tm_{i} &= \tm_{i-3} \ \flip{\tm_{i-4}} \ 
\tm_{i-4} \ \flip{\tm_{i-3}} \ 
\tm_{i-2} \ 
\tm_{i-4} \ \flip{\tm_{i-3}} \ 
\flip{\tm_{i-4}} \ \flip{\tm_{i-3}} 
\label{eq:tm-fac-i-4-flip-i-3} 
\\
\tm_{i} &=  \tm_{i-3} \ 
\flip{\tm_{i-4}} \ \tm_{i-3} \ 
\tm_{i-4} \ \tm_{i-2} \ \tm_{i-4} \ 
\flip{\tm_{i-4}} \ \tm_{i-3} \  
\flip{\tm_{i-3}}
\label{eq:tm-fac-flip-i-4-i-3}
\end{align}
\end{restatable}

\begin{figure}[t]
\centering
\includegraphics[width=0.55\linewidth]{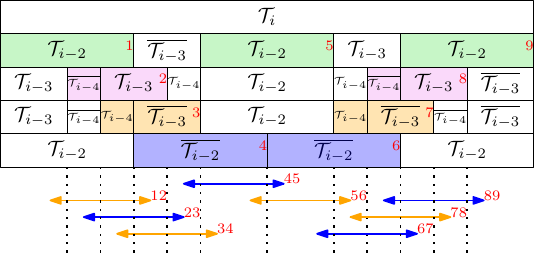}
\caption{
An illustration of several factorizations of $\tm_i$ from \cref{thm:tm-i-fac}.
Net occurrences of each string in \cref{def:tm-pos-nf-str} 
are highlighted in a separate color.
Super-occurrences of the eight BNSOs are shown as colored arrows,
blue for $\tm_{i-3}$ and orange for $\flip{\tm_{i-3}}$ (see \cref{thm:ext-tm-i-3}).
Each net occurrence is numbered in red
at the top-right corner,
and each arrow with label $ij$
corresponds to an overlap between 
the $i^\text{th}$ and $j^\text{th}$ net occurrences.
}
\label{fig:tm-non-net-occ-cases}
\end{figure}

The following two results immediately follow from \cref{thm:tm-occ}.
Note that \cref{thm:tm-occ-234} also appears in \cite{journal/jda/2012/radoszewski}.

\begin{corollary}\label{thm:tm-occ-234}
For each $i \geq 5$:
\begin{itemize}
\item 
$\tm_{i-2}$ only occurs at positions
1, $\ltm_{i-2} + \ltm_{i-3} +1$ and $\ltm_{i-1} + \ltm_{i-2} +1$
in $\tm_i$.
\item 
$\flip{\tm_{i-2}}$ only occurs at positions
$\ltm_{i-2} +1$ and  $\ltm_{i-1} +1$
in $\tm_i$.
\item 
$\tm_{i-4} \ \flip{\tm_{i-3}}$ only occurs at positions 
$\ltm_{i-3} + \ltm_{i-4} +1$ and  $\ltm_{i-1} + \ltm_{i-3} + 1$
in $\tm_i$.
\item 
$\flip{\tm_{i-4}} \ \tm_{i-3}$ only occurs at positions 
$\ltm_{i-3} +1 $ and $  \ltm_{i-1} + \ltm_{i-3} + \ltm_{i-4} + 1$
in $\tm_i$.
\end{itemize}
\end{corollary}

\begin{corollary}\label{thm:tm-i-3-occ}
$\tm_{i-3}$ only occurs at positions
$1, 
\ltm_{i-3} + \ltm_{i-4} + 1,  
\ltm_{i-2} + \ltm_{i-3} + 1,  
\ltm_{i-1} + \ltm_{i-3} + 1, \text{~and~}
\ltm_{i-1} + \ltm_{i-2} + 1
$
in $\tm_{i}$.
\end{corollary}

We now identify the nine net occurrences in $\tm_i$.

\begin{lemma}\label{thm:tm-net-occ}
Each occurrence of each string in $\mathcal{P}_i$ is a net occurrence in $\tm_i$.
\end{lemma}

\begin{proof}
We proceed by examining the left and right extension characters
of each occurrence of each string in $\mathcal{P}_i$.

Since $\tm_{i-2}$ is a prefix and a suffix of $\tm_i$,
by the definition of occurrences,
the occurrence of $\tm_{i-2}$ at  positions 1 has a unique left extension character,
and the occurrence of $\tm_{i-2}$ at  positions $\ltm_{i-1} + \ltm_{i-2} + 1$ has a unique right extension character.
Next, note that the right extension character of 
the occurrence of $\tm_{i-2}$ at position 1 
differs from that of 
the occurrence at position $\ltm_{i-2} + \ltm_{i-3} + 1$
because $\flip{\tm_{i-3}}[1] \neq \tm_{i-3}[1]$.
Similarly, 
the left extension character of 
the occurrence at position $\ltm_{i-2} + \ltm_{i-3} +1$
differs from that of
the occurrence at position $\ltm_{i-1} + \ltm_{i-2} +1$,
because $\flip{\tm_{i-3}}[\ltm_{i-3}] \neq \tm_{i-3}[\ltm_{i-3}]$.
Hence, all three occurrences of $\tm_{i-2}$ are net occurrences.

For $\flip{\tm_{i-2}}$, a similar argument holds.
the right extension characters satisfy
$\flip{\tm_{i-2}}[1] \neq \tm_{i-2}[1]$
and the left extension characters satisfy
$\tm_{i-2}[\ltm_{i-2}] \neq \flip{\tm_{i-2}}[\ltm_{i-2}]$.
Thus, both occurrences of $\flip{\tm_{i-2}}$
are net occurrences.
For $\tm_{i-4} \ \flip{\tm_{i-3}}$, 
similarly,
the right extension characters satisfy
$\tm_{i-2}[1] \neq \flip{\tm_{i-4}}[1]$
and the left extension characters satisfy
$\flip{\tm_{i-4}}[\ltm_{i-4}] = \flip{\tm_{i-2}}[\ltm_{i-2}] 
\neq \tm_{i-2}[\ltm_{i-2}]$.
Thus, both occurrences of $\tm_{i-4} \ \flip{\tm_{i-3}}$
are net occurrences.
Finally, for $\flip{\tm_{i-4}} \ \tm_{i-3}$,
once again,
the right extension characters satisfy
$\tm_{i-4}[1] \neq \flip{\tm_{i-3}}[1]$
and the left extension characters satisfy
$\tm_{i-3}[\ltm_{i-3}] \neq \tm_{i-4}[\ltm_{i-4}]$.
Thus, both occurrences of $\flip{\tm_{i-4}} \ \tm_{i-3}$
are net occurrences.
\end{proof}

To show that all other occurrences are not net occurrences,
we follow \cref{thm:super-occ-bnso}.
First note that the nine net occurrences 
we identified in \cref{thm:tm-net-occ} form an ONOC of $\tm_i$.
The eight BNSOs of this ONOC correspond to 
the occurrences of 
$\tm_{i-3}$ and $\flip{\tm_{i-3}}$ shown in \cref{fig:tm-non-net-occ-cases}.
We next show that no super-occurrences of these occurrences 
are net occurrences to conclude that this ONOC already contains
all the net occurrences in $\tm_i$.

\begin{lemma}\label{thm:ext-tm-i-3}
Consider an occurrence $(s, e)$  in $\tm_i$
and let $S := \tm_i[s \ldots e]$.
If $(s, e)$ is a proper super-occurrence  
of $\tm_{i-3}$ or $\flip{\tm_{i-3}}$,
 and $S \notin \mathcal{P}_i$,
then $(s, e)$ is not a net occurrence.
\end{lemma}

\begin{proof}
We first consider when $(s, e)$ contains an occurrence of $\tm_{i-3}$.
Consider strings $X$ and $Y$ such that $S = X \ \tm_{i-3} \ Y$ 
and  $X \ Y \neq \epsilon$.
Let position $ u := s + |X| $ be the starting position of this occurrence of $\tm_{i-3}$.
Let $C := \{ 
1, \ltm_{i-2} + \ltm_{i-3} + 1, \ltm_{i-1} + \ltm_{i-2} + 1
\}$
and 
$D := \{
\ltm_{i-3} + \ltm_{i-4} + 1, \ltm_{i-1} + \ltm_{i-3} + 1
\}$.
By \cref{thm:tm-i-3-occ}, 
we have $u \in C \cup D$.
We next examine the following cases depending on 
which set $u$ belongs to and how large $|Y|$ is.

We first consider when $u \in C$.
\begin{enumerate}[(a)]
\item\label{tm-case-1}
$|Y| < \ltm_{i-3}$.
By \cref{thm:tm-occ-234},
note that $Y$ is always a prefix of $\flip{\tm_{i-3}}$. 
This means the right extension character of $S$ is always $\flip{\tm_{i-3}}[|Y|+1]$.
Thus, $(s, e)$ is not  a net occurrence.
\item\label{tm-case-2}
$|Y| = \ltm_{i-3}$.
By \cref{thm:tm-occ-234}, $ S \in \mathcal{P}_i$ in this case.
\item\label{tm-case-3}
$|Y| > \ltm_{i-3}$.
By \cref{thm:tm-occ-234}, $(s, e)$ contains a net occurrence of $\tm_{i-3} \ \flip{\tm_{i-3}} = \tm_{i-2}$ as a proper sub-occurrence.
By \cref{thm:non-net-super} and \cref{thm:tm-net-occ}, $(s, e)$ is not  a net occurrence.
\end{enumerate}

We next consider when $u \in D$.
\begin{enumerate}[(a)]
\item\label{tm-case-4}
$|Y| < \ltm_{i-4}$.
Recall that  $ \tm_{i-4} \ \flip{\tm_{i-3}}
= \tm_{i-4} \ \flip{\tm_{i-4}} \ \tm_{i-4}
= \tm_{i-3} \ \tm_{i-4}$.
By \cref{thm:tm-occ-234},
note that $Y$ is always a prefix of $\tm_{i-4}$ in this case. 
This means the right extension character of $S$ is always $\tm_{i-4}[|Y|+1]$.
Thus, $(s, e)$ is not  a net occurrence.
\item\label{tm-case-5}
$|Y| = \ltm_{i-4}$. 
Using \cref{thm:tm-occ-234}, $ S \in \mathcal{P}_i$ in this case.
\item\label{tm-case-6}
$|Y| > \ltm_{i-4}$. By \cref{thm:tm-occ-234}, in this case $(s, e)$ contains a net occurrence of $\tm_{i-3} \ \tm_{i-4} = \tm_{i-4} \ \flip{\tm_{i-4}} \ \tm_{i-4} = \tm_{i-4} \ \flip{\tm_{i-3}}$ as a proper sub-occurrence.
Thus, by \cref{thm:non-net-super} and \cref{thm:tm-net-occ}, $(s, e)$ is not  a net occurrence.
\end{enumerate}
We can prove the case when $(s, e)$ contains an occurrence of $\flip{\tm_{i-3}}$ similarly.
\end{proof}

Finally, the main result follows from \cref{thm:super-occ-bnso}, \cref{thm:tm-net-occ} and \cref{thm:ext-tm-i-3}.

\begin{theorem}\label{thm:tm-only-net-occ}
The net occurrences in \cref{thm:tm-net-occ}
are the only net occurrences in each $\tm_i$.
\end{theorem}

\section{Conclusion and Future Work}

In this work, we investigate net occurrences in Fibonacci and Thue-Morse words, 
making two main contributions.
First, we confirm the conjecture that each Fibonacci word contains
exactly three net occurrences.
Second, we establish that each Thue-Morse word contains exactly nine net occurrences. 
To achieve these results, we first introduce the notion of
an overlapping net occurrence cover 
and show how it can be used to 
prove that certain net occurrences in a text are the only ones. 
We then develop recurrence relations 
that precisely characterize the occurrences of 
Fibonacci and Thue-Morse words of smaller order,
which could be of independent interest.
As an application, we illustrate how these results
facilitate the counting of small-order occurrences.

An avenue of future work is to extend our findings to 
study the net occurrences in $k$-bonacci words~\cite{conf/cwords/2023/gheeraert,
journal/tcs/2022/jahannia,
journal/tcs/2021/ghareghani,
journal/combinatorics/2020/ghareghani} and Thue-Morse-like words~\cite{journal/tcs/2022/aedo,journal/dam/2019/chen}.
Furthermore, since both Fibonacci and Thue-Morse words can be defined via morphisms,
one could also explore net occurrences in other morphic words~\cite{conf/dlt/2022/frosini,conf/dlt/2010/halava,conf/iwoca/2019/brlek}.
Finally, the net occurrences have been characterized in terms of 
minimal unique substrings~\cite{conf/cpm/2025/mieno};
this viewpoint may offer alternative and potentially simpler proofs than those presented
in Sections \ref{sec:fib-net-occ}--\ref{sec:tm-net-occ}.

\clearpage

\bibliographystyle{plainurl}
\bibliography{references}

\appendix

\section{Proofs Omitted from 
\texorpdfstring{\cref{sec:onoc}}{}
}\label{appx:onoc}

\nonnetsuper*

\begin{proof}
Let $(s', e')$ be the net occurrence, 
then $T[s'-1 \ldots e']$ and $T[s' \ldots e'+1]$ are both unique. 
Since $T[s \ldots e]$ contains at least one of these two strings as a substring,
$T[s \ldots e]$ is also unique.
Thus, $(s,e)$ is not a net occurrence.  
\end{proof}

\nonnetsub*

\begin{proof}
Let $(s', e')$ be the net occurrence, 
then $T[s' \ldots e']$ is repeated.
Since $(s,e)$ is a proper sub-occurrence of $(s', e')$,
$T[s-1 \ldots e]$ or $T[s \ldots e+1]$ is also repeated.
Thus, $(s,e)$ is not a net occurrence. 
\end{proof}

\superoccbnso*

\begin{proof}
Let $(s,e)$ be a net occurrence in $T$ that is outside of $\mathcal{C}$. 
Assume, by contradiction, that 
$(s,e)$ is not a super-occurrence of any occurrence $(i-1,j+1)$, 
where $(i,j)$ is an occurrence in the set of BNSOs, 
$ \{ (i_2, j_1), (i_3, j_2), \ldots (i_{c}, j_{c-1}) \}$.
We consider the following two cases depending on the position of  $s$.

First, when
$i_{k+1} \leq s < i_{k+2}$ for some $0 \leq k \leq c-2$.
Given our assumption,
$(s,e)$ is not a super-occurrence of  $(i_{k+2}-1, j_{k+1} + 1)$,
where $(i_{k+2}, j_{k+1})$ is a BNSO.
It follows that $e$ satisfies  $e \leq j_{k+1}$,
implying that $(s,e)$ must be a sub-occurrence of $(i_{k+1}, j_{k+1})$, 
which is a net occurrence in $\mathcal{C}$. 
Now,
if $(s,e)$ is a proper sub-occurrence of $(i_{k+1}, j_{k+1})$, this
this contradicts \cref{thm:non-net-sub};
if $(s,e)$ is $(i_{k+1}, j_{k+1})$,  it contradicts the assumption that 
$(s,e)$ is a net occurrence in $T$ outside of $\mathcal{C}$.
Second, when 
$s \geq i_c$.
In this case, $(s,e)$ is a sub-occurrence of $(i_c, j_c)$,
the last net occurrence in $\mathcal{C}$.

In both cases, $(s,e)$ must be a sub-occurrence of some net occurrence in $\mathcal{C}$,
leading to a contradiction of either the assumption or \cref{thm:non-net-sub}.
Therefore, 
our initial assumption was false and we conclude that
$(s,e)$ is indeed a super-occurrence of $(i-1,j+1)$, 
where $(i,j)$ is a BNSO of $\mathcal{C}$.
\end{proof}

\section{Proofs Omitted from 
\texorpdfstring{\cref{sec:fib-net-occ}}{}
}\label{appx:fib-net-occ-proofs}

\fitwoqiocc*

\begin{proof}
By \cref{thm:f-i-2}, there are only three positions where 
$F_{i-2} \ Q_i$ could occur.
First observe that $F_{i-2} \ Q_i$ cannot occur at position $f_{i-1}+1$.
Next, by \cref{thm:f-i-fac} and \cref{thm:q-i-delta},
the occurrences of $F_{i-2}$ at positions 1 and $f_{i-2}+1$ are both 
followed by $Q_i$,
thus $F_{i-2} \ Q_i$ only occurs at these two positions.
\end{proof}

\qlen*

\begin{proof}
Note that
$
|Q_i| 
= \sum_{j=2}^{i-5} |F_j| 
= \left( \sum_{j=1}^{i-5} f_j \right) - f_1
= f_{i-3} - 1 - f_1
= f_{i-3} - 2
$
where the third equality comes from
the fact that $\sum_{j=1}^k f_j = f_{k+2} - 1$.
\end{proof}

\fitwoqi*

\begin{proof}
By \cref{thm:f-i-2-q-i-occ} and \cref{thm:net-occ-fib},
$F_{i-2} \ Q_i$ only occurs twice in $F_i$,
and both are net occurrences, 
which means the extensions are unique.
Thus, any string containing $F_{i-2} \ Q_i$
as a substring is also unique.
\end{proof}

\qione*

\begin{proof}
Observe that the occurrence of $F_{i-3}$ at position $f_{i-3}+1$
is followed by $F_{i-6}\ F_{i-5} = Q_{i-1} \ \Delta(1 - (i \bmod 2))$ while 
the other three occurrences of $F_{i-3}$
are all followed by $F_{i-4} = F_{i-5}\ F_{i-6} = Q_{i-1} \ \Delta(i \bmod 2)$.
\end{proof}

\fitsf*

\begin{proof}
From the proof of \cref{thm:q-i-1},
$F_{i-3}$ is only followed by $F_{i-6}\ F_{i-5}$ once
and by $F_{i-4}$ three times,
thus $F_{i-3} \ F_{i-6} \ F_{i-5}$ is unique.
Next, by \cref{thm:q-len},
$|F_{i-3} \ Q_{i-1}| = f_{i-3} + f_{i-4} - 2 = f_{i-2} - 2$.
Also from the proof of \cref{thm:q-i-1},
$F_{i-3} \ Q_{i-1}$ is only followed by 
$\Delta(1 - (i \bmod 2))$ once and by 
$\Delta(i \bmod 2)$ three times,
thus, the length-$(f_{i-2}-1)$ prefix of 
$F_{i-3} \ Q_{i-1} \ \Delta(1 - (i \bmod 2)) = F_{i-3} \ F_{i-6} \ F_{i-5}$ is also unique.
\end{proof}

\fithreerightext*

\begin{proof} 
Let $U$ be the length-$(f_{i-3}-1)$ prefix of $F_{i-3}$.
By \cref{eq:fac-1},
$F_{i-3} = Q_i \ \Delta(1 - (i \bmod 2)$.
Note that $U[|U|-1]$ is always \texttt{a} 
because $Q_i$ ends with $F_2= \texttt{a}$.
When $U[|U|] = \texttt{a}$,
 the right extension character of $U$ is always \texttt{b}.
This is because, if it were not, \texttt{aaa} would occur in $F_i$,
contradicting \cref{thm:no-aaa}.
On the other hand, when
$U[|U|] = \texttt{b}$,
the right extension character of $U$ is always \texttt{a}.
This is because, similarly, an occurrence of \texttt{bb} would contradict \cref{thm:no-aaa}.
Finally, observe that, whether $U[|U|]$ is \texttt{a} or \texttt{b},
$U[|U|]$ concatenated with the right extension character of $U$ 
is exactly $\Delta(1 - (i \bmod 2)$.
Therefore, the desired result follows.
\end{proof}

\section{Proof Omitted from 
\texorpdfstring{\cref{sec:tm-net-occ}}{}
}\label{appx:tm-net-occ-proofs}

\tmifac*

\begin{proof}
The first two follow from 
\cref{eq:tm-fac-i-2}  and \cref{eq:tm-fac-flip-i-2}.
We next proceed by repeatedly applying 
the definition of Thue-Morse words.
Substituting 
$
  \tm_{i-2}
= \tm_{i-3} \ \flip{\tm_{i-3}} 
= \tm_{i-3} \ \flip{\tm_{i-4}} \ \tm_{i-4} 
$
and 
$
  \tm_{i-3} \ \tm_{i-2}
= \tm_{i-4} \ \flip{\tm_{i-4}} \ \tm_{i-3} \ \flip{\tm_{i-3}}
= \tm_{i-4} \ \flip{\tm_{i-4}} \ \tm_{i-4} \ \flip{\tm_{i-4}} \ \flip{\tm_{i-3}}
= \tm_{i-4} \ \flip{\tm_{i-3}} \ \flip{\tm_{i-4}} \ \flip{\tm_{i-3}}
$
to \cref{eq:tm-fac-i-2},
we have \cref{eq:tm-fac-i-4-flip-i-3}.
Finally, observe that 
$
  \tm_{i-4} \ \flip{\tm_{i-3}}
= \tm_{i-4} \ \flip{\tm_{i-4}} \ \tm_{i-4}
= \tm_{i-3} \ \tm_{i-4}
$
and similarly,
$
  \flip{\tm_{i-3}} \ \flip{\tm_{i-4}}
= \flip{\tm_{i-4}} \ \tm_{i-4} \ \flip{\tm_{i-4}}
= \flip{\tm_{i-4}} \ \tm_{i-3}
$.
Substituting them to \cref{eq:tm-fac-i-4-flip-i-3},
we derive \cref{eq:tm-fac-flip-i-4-i-3}.
\end{proof}

\section{A Factorization of Thue-Morse Word}\label{sec:tm-fac}

First, we define a \emph{smallest} factorization of a string as 
one that contains the fewest number of factors while satisfying certain conditions. 
In this section, we explore two smallest factorizations of 
$\tm_i$, each containing all occurrences of $\tm_{i-j}$ and $\flip{\tm_{i-j}}$, respectively.
Observe that such factorizations exist due to 
the overlap-free property of each Thue-Morse word (\cref{thm:tm-overlap-free}).

\begin{definition}\label{def:tm-fac}
For each $i \geq 2$ and $0 \leq j \leq i-1$, 
we define the following.
\begin{itemize}
\item 
Let $\fac{A}{i,j}$ and $\fac{B}{i,j}$
denote the smallest factorization of $\tm_i$ that contains 
all occurrences of $\tm_{i-j}$ and $\flip{\tm_{i-j}}$
in $\tm_i$, respectively.
\item 
Define $\nextfac{T_{i-j}} := T_{i-(j+1)}$ and $\nextfac{\overline{T_{i-j}}} := \overline{T_{i-(j+1)}}$.
\item 
Consider $\fac{}{i,j} \in \{ \fac{A}{i,j}, \fac{B}{i,j} \}$  
and suppose
 $\fac{}{i,j} = (  x_t )^m_{t=1}$.
We define two operators on $\fac{}{i,j}$:
\[
\nextfac{\fac{}{i,j}} := \left(  \nextfac{x_t} \right)^m_{t=1}
\quad \text{and} \quad
\flip{\fac{}{i,j}} := \left(  \flip{x_t} \right)^m_{t=1}.
\]
\item 
Consider two factorizations $\mathcal{X} = (x_k)^{m}_{k=1}$
and $\mathcal{Y} = (y_k)^{\ell}_{k=1}$.
If $x_{m} = y_{1} = \flip{\tm_{i-j}}$,
then 
\[ \mathcal{X} \boxplus \mathcal{Y} := \left( 
x_1, \ x_2, \ \ldots, \ x_{m-1},   \
\flip{\tm_{i-(j+1)}}, \ \tm_{i-j}, \ \tm_{i-(j+1)}, \
y_2, \ y_3, \ \ldots, \ y_{\ell}
\right). \lipicsEnd \]
\end{itemize}
\end{definition}

\noindent
For the definition of operator $\boxplus$,
note that  $|\mathcal{X} \boxplus \mathcal{Y}| = |\mathcal{X}| + |\mathcal{Y}| + 1$ and
\[
\flip{\tm_{i-j}} \  \flip{\tm_{i-j}}
= \flip{\tm_{i-(j+1)}} \ \tm_{i-(j+1)} \ \flip{\tm_{i-(j+1)}} \ \tm_{i-(j+1)}
= \flip{\tm_{i-(j+1)}} \ \tm_{i-j} \ \tm_{i-(j+1)}.
\]

We next introduce  a simple characteristic of $\fac{A}{i,j}$ and $\fac{B}{i,j}$.

\begin{observation}\label{thm:tm-fac-no-two-consec}
For each $i \geq 2$ and $0 \leq j \leq i-1$,
consider a factorization $\mathcal{X} = (  x_k )^m_{k=1}$ of $\tm_i$
that contains all occurrences of $\tm_{i-j}$ (respectively, $\flip{\tm_{i-j}}$).
If no two consecutive factors are both different from $\tm_{i-j}$ (respectively, $\flip{\tm_{i-j}}$),
then 
$\mathcal{X}$ is the smallest  and thus
$\mathcal{X} = \fac{A}{i,j}$ (respectively, $\mathcal{X} = \fac{B}{i,j}$).
\end{observation}

The observation 
holds because, otherwise, 
we could merge the two consecutive factors 
and obtain a smaller factorization.

Now, we present the main result of the section,
illustrated in \cref{fig:tm-occ-fac-evolve}.

\begin{figure}[t]
\centering
\includegraphics[width=0.8\linewidth]{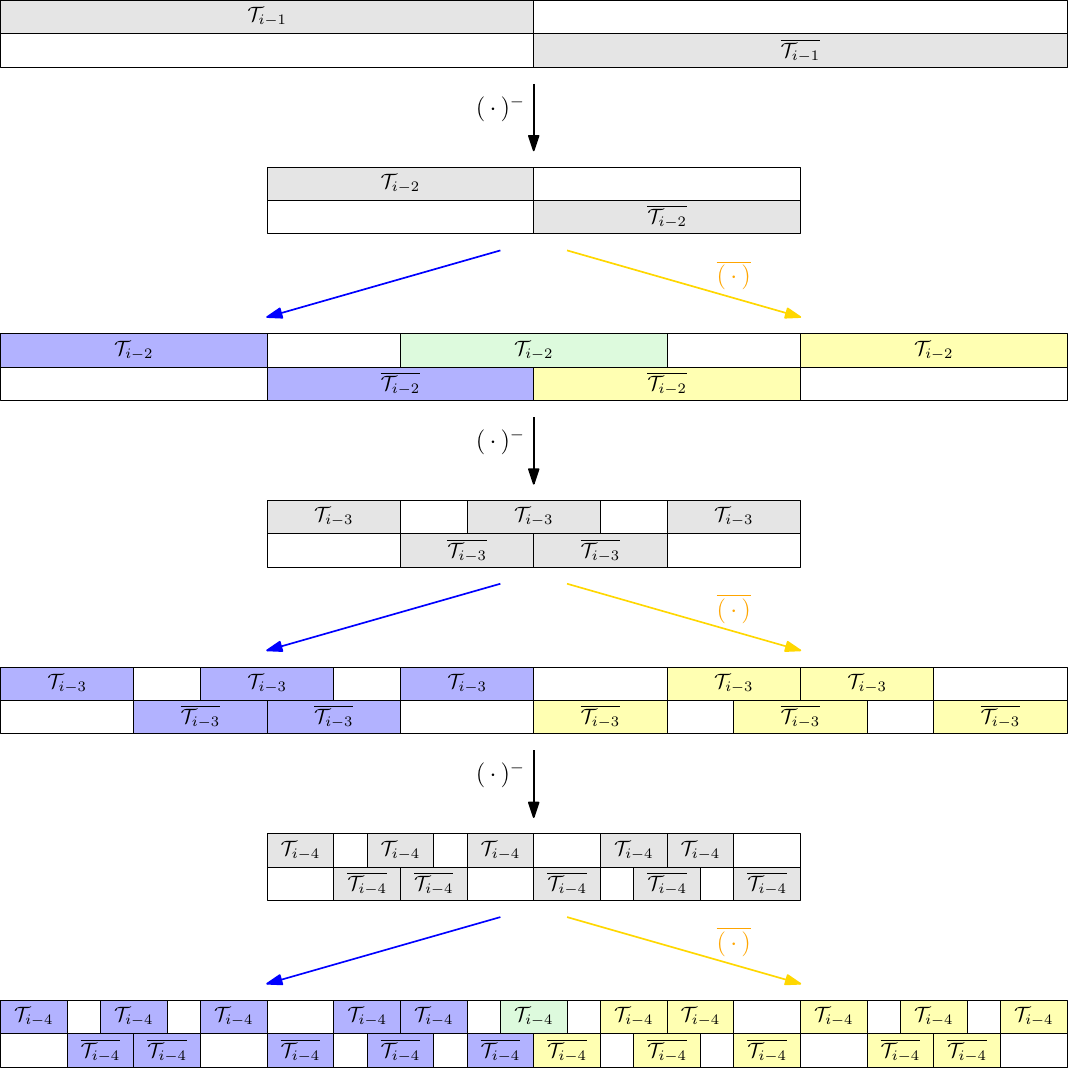}
\caption{
An illustration of \cref{thm:tm-fac} for  $1 \leq j \leq 4$.
Operators $(\cdot)^-$ and $\flip{(\cdot)}$ are defined in 
\cref{def:tm-fac}.
Notice the green occurrences of $\tm_{i-j}$ are introduced from $\boxplus$.
}
\label{fig:tm-occ-fac-evolve}
\end{figure}

\begin{theorem}\label{thm:tm-fac}
For  $i \geq 2$ and $0 \leq j \leq i-1$,
the following statements hold.
\begin{enumerate}[(1)]
\item\label{statment-1} 
For each $\fac{}{i,j} \in \{ \fac{A}{i,j}, \fac{B}{i,j} \}$,
let $\fac{}{i,j} = (  x_t )^m_{t=1}$.
Then each term $x_t$ is in the set
$\mathcal{B}_{i, j} := \left\{ 
\tm_{i-j}, \flip{\tm_{i-j}}, \tm_{i-(j+1)}, \flip{\tm_{i-(j+1)}} 
\right\}$.
Moreover, $x_1 = \tm_{i-j}$.
If $j$ is even,
then $x_m = \tm_{i-j}$; otherwise, $x_m = \flip{\tm_{i-j}}$.
\item\label{statment-2} 
$\fac{A}{i,0} = \left( \tm_{i} \right)$,
$\fac{A}{i,1} = \left( \tm_{i-1}, \flip{\tm_{i-1}}  \right)$,
and for each $j \geq 2$,
\[
\fac{A}{i,j} = 
\begin{cases}
\nextfac{\fac{A}{i,j-1}} \ \flip{\nextfac{\fac{B}{i,j-1}}},  & \quad    \text{$j$ is odd},\\
\nextfac{\fac{A}{i,j-1}} \boxplus \flip{\nextfac{\fac{B}{i,j-1}}}, &  \quad   \text{$j$ is even}.
\end{cases}
\]
\item\label{statment-3} 
$\fac{B}{i,0} = \emptyseq$, 
$\fac{B}{i,1} = \left( \tm_{i-1}, \flip{\tm_{i-1}}  \right)$,
and for each $j \geq 2$,
$
\fac{B}{i,j} = \nextfac{\fac{B}{i,j-1}} \ \flip{\nextfac{\fac{A}{i,j-1}}} .
$ 
\end{enumerate}
\end{theorem}

\begin{proof}

We proceed by induction on $j$.

\textbf{Base cases.}
The claim holds trivially for  $j=0$.
When $j= 1$,
$\left( \tm_{i-1}, \flip{\tm_{i-1}}  \right)$
is the smallest factorization following \cref{thm:tm-fac-no-two-consec} and 
Statement~\ref{statment-1} holds.
Next, note that 
$\nextfac{\fac{A}{i,1}}  = \nextfac{\fac{B}{i,1}} 
=  \left( \tm_{i-2}, \flip{\tm_{i-2}}  \right)$.
When $j=2$, it follows that
\begin{equation}\label{eq:tm-fac-i-2}  
\fac{A}{i,2}
= \nextfac{\fac{A}{i,1}} \boxplus  \flip{\nextfac{\fac{B}{i,1}}} 
= \left(  \tm_{i-2}, \flip{\tm_{i-2}} \right) \boxplus \left(  \flip{\tm_{i-2}}, \tm_{i-2} \right)
= \left(  \tm_{i-2}, \flip{\tm_{i-3}}, \tm_{i-2}, \tm_{i-3}, \tm_{i-2} \right)
\end{equation}
and  
\begin{equation}\label{eq:tm-fac-flip-i-2} 
\fac{B}{i,2}
= \nextfac{\fac{B}{i,1}}   \flip{\nextfac{\fac{A}{i,1}}} 
= \left(  \tm_{i-2}, \flip{\tm_{i-2}} \right) \ \left( \flip{\tm_{i-2}}, \tm_{i-2} \right)
= \left(  \tm_{i-2}, \flip{\tm_{i-2}},  \flip{\tm_{i-2}}, \tm_{i-2} \right)
\end{equation}
Both factorizations are the smallest 
following \cref{thm:tm-fac-no-two-consec}, and Statement~\ref{statment-1} holds.

\textbf{Inductive step.}
Let  $k$ be an odd integer such that $3 \leq k \leq i-1$, and
assume the claim holds for $j=k-1$.
Specifically,  assume $\fac{A}{i,k-1} = (x_t)^m_{t=1}$, 
$\fac{B}{i, k-1} = (  y_t )^l_{t=1}$,
$x_1 = y_1 = \tm_{i-(k-1)}$, and
$x_m  = y_l = \tm_{i-(k-1)}$.
We now prove the result for $j = k$. 

First, note that both $\fac{A}{i,k-1}$ and $\fac{B}{i,k-1}$ are factorizations of $\tm_{i}$.
Then, by the definition of operation $\nextfac{\cdot}$,
both $\nextfac{\fac{A}{i,k-1}}$ and $\nextfac{\fac{B}{i,k-1}}$ are factorizations of $\tm_{i-1}$,
and  $\flip{\nextfac{\fac{B}{i,k-1}}}$ is a factorization of $\flip{\tm_{i-1}}$.
Now, since  $\tm_i = \tm_{i-1} \ \flip{ \tm_{i-1}}$,
it follows that
$\mathcal{Y}_{\text{odd}}^A := \nextfac{\fac{A}{i,k-1}} \ \flip{\nextfac{\fac{B}{i,k-1}}}$ 
is a factorization of $\tm_i$.
It remains to show that $\mathcal{Y}_{\text{odd}}^A = \fac{A}{i,j}$. 
Observe that 
\begin{equation}\label{eq:y-a-odd}
\mathcal{Y}_{\text{odd}}^A 
= \nextfac{\fac{A}{i,k-1}} \ \flip{\nextfac{\fac{B}{i,k-1}}} 
= \nextfac{x_1} \ \cdots \ \nextfac{x_{m-1}} \ \tm_{i-k} \ \flip{\tm_{i-k}} \ \flip{\nextfac{y_2}} \ \cdots \ \flip{\nextfac{y_l}}.
\end{equation}
By the induction hypothesis on $\fac{A}{i,k-1}$ and $\fac{B}{i,k-1}$ 
and the definition of operation $\nextfac{\cdot}$, we know that 
$\nextfac{\fac{A}{i,k-1}}$ contains all the occurrences of $\tm_{i-k}$ in $ \tm_{i-1}$,
and $\flip{\nextfac{\fac{B}{i,k-1}}}$ contains all the occurrences of $\tm_{i-k}$ in $\flip{ \tm_{i-1}}$. 
Since  $\tm_i = \tm_{i-1} \ \flip{ \tm_{i-1}}$
and $\tm_i$ is overlap-free,
it follows that $\mathcal{Y}_{\text{odd}}^A$ contains all the occurrences of $\tm_{i-k}$.
Moreover,  no two consecutive factors of $\mathcal{Y}_{\text{odd}}^A$ are both different from $\tm_{i-k}$,
so by \cref{thm:tm-fac-no-two-consec},
we conclude that 
$\mathcal{Y}_{\text{odd}}^A = \fac{A}{i,j}$. 

Next we show that all factors of $\fac{A}{i,j}$ are elements of $ \mathcal{B}_{i, k}$.  
Since $x_t \in \mathcal{B}_{i, k-1}$ for each $1 \leq t \leq m$ 
and
$y_t \in \mathcal{B}_{i, k-1}$ for each $1 \leq t \leq l$,
it follows from \cref{eq:y-a-odd} that
all factors of $\fac{A}{i,j}$  are elements of $\mathcal{B}_{i, k}$. 
Additionally, the first factor in $\fac{A}{i,j}$ is $\nextfac{x_1} = \tm_{i-k}$, 
and the last factor in $\fac{A}{i,j}$ is $\flip{\nextfac{y_l}} = \flip{\tm_{i-k}}$ since $k-1$ is even.

Similarly, we can show that $\mathcal{Y}_{\text{odd}}^B := \nextfac{\fac{B}{i,k-1}} \ \flip{\nextfac{\fac{A}{i,k-1}}}$
is a factorization  of $\tm_i$,
 that $\mathcal{Y}_{\text{odd}}^B = \fac{B}{i,j}$,
and that Statement~\ref{statment-1} holds.

We can prove analogously when $k$ is even.
In this case, 
 the operation $\boxplus$ is used to ensure that $\fac{A}{i,j}$ contains all occurrences of $\tm_{i-j}$.
Specifically, when $k$ is even, we have $\nextfac{x_m}= \flip{\nextfac{y_1}} = \flip{\tm_{i-k}}$, 
and there is a occurrence of $\tm_{i-k}$ within $\flip{\tm_{i-k}} \ \flip{\tm_{i-k}} =\flip{\tm_{i-(k+1)}} \ \tm_{i-k} \ \tm_{i-(k+1)}$.
\end{proof}

\end{document}